\documentclass[a4paper,UKenglish]{lipics-v2018}

\usepackage{amssymb,amsmath,amsthm}
\usepackage{microtype}
\usepackage[small,basic]{complexity}
\usepackage{multirow}
\usepackage{tikz}
\usetikzlibrary{positioning,automata,arrows}
\usetikzlibrary{decorations.pathreplacing,decorations.pathmorphing,decorations.text}
\usetikzlibrary{shapes.geometric,patterns,backgrounds}

\usepackage{color,colortbl}
\usepackage{extarrows}
\usepackage{mathtools}
\usepackage[strings]{underscore}
\usepackage{slashbox}
\usepackage{bm}
\usepackage{cite}
\usepackage{xspace}
\usepackage{xifthen}
\usepackage{multicol}
\usepackage[shortlabels]{enumitem}
\usepackage{subcaption}
\usepackage{oubraces}

\usepackage{hyperref}

\newcolumntype{L}[1]{>{\centering\arraybackslash}m{#1}}

\newtheorem{problem}[theorem]{Problem}

\newclass{\EXPTIME}{EXPTIME}

\newcommand{\Q}{\mathbb{Q}}
\newcommand{\Z}{\mathbb{Z}}
\newcommand{\N}{\mathbb{N}}

\renewcommand{\PCP}{\rm PCP\xspace}

\newcommand{\FG}[1][\Gamma]{{\rm FG}(#1)}
\newcommand{\ba}{\mathbf a}
\newcommand{\bb}{\mathbf b}

\newcommand{\bv}{\mathbf v}

\newcommand{\by}{\mathbf y}

\newcommand{\bu}{\mathbf u}
\newcommand{\overbar}[1]{\mkern 1.5mu\overline{\mkern-1.5mu#1\mkern-1.5mu}\mkern 1.5mu}

\newcommand{\gltwoz}{{\rm GL}\ensuremath{(2,\mathbb{Z})}\xspace}
\newcommand{\sltwoz}{{\rm SL}\ensuremath{(2,\mathbb{Z})}\xspace}
\newcommand{\slfourz}{{\rm SL}\ensuremath{(4,\mathbb{Z})}\xspace}
\newcommand{\slthreez}{{\rm SL}\ensuremath{(3,\mathbb{Z})}\xspace}
\newcommand{\slthreek}{{\rm SL}\ensuremath{(3,\mathbb{K})}\xspace}
\newcommand{\slthreeq}{{\rm SL}\ensuremath{(3,\mathbb{Q})}\xspace}

\newcommand{\slthreec}{{\rm SL}\ensuremath{(3,\mathbb{C})}\xspace}
\newcommand{\heis}{{\rm H}\ensuremath{(3,\mathbb{Z})}\xspace}
\newcommand{\heisr}{{\rm H}\ensuremath{(3,\mathbb{R})}\xspace}
\newcommand{\heisq}{{\rm H}\ensuremath{(3,\mathbb{Q})}\xspace}
\newcommand{\heisnq}{{\rm H}\ensuremath{(n,\mathbb{Q})}\xspace}
\newcommand{\heisc}{{\rm H}\ensuremath{(3,\mathbb{C})}\xspace}
\newcommand{\HQS}{$\mathbb{H}(\mathbb{Q})^{2\times 2}$\xspace}
\newcommand{\heisk}{{\rm H}\ensuremath{(3,\mathbb{K})}\xspace}
\newcommand{\psltwoz}{{\rm PSL}\ensuremath{(2,\mathbb{Z})}\xspace}
\DeclareMathOperator{\tr}{tr}

\title{On the Identity Problem for the Special Linear Group and the Heisenberg Group}

\author{Sang-Ki Ko}{Korea Electronics Technology Institute, South Korea}{sangkiko@keti.re.kr}{}{}
\author{Reino Niskanen}{Department of Computer Science, University of Liverpool, UK}{r.niskanen@liverpool.ac.uk}{}{}
\author{Igor Potapov}{Department of Computer Science, University of Liverpool, UK}{potapov@liverpool.ac.uk}{}{}

\authorrunning{S-K. Ko, R. Niskanen and I.Potapov}

\Copyright{Sang-Ki Ko, Reino Niskanen and Igor Potapov}

\subjclass{, , }


\category{}

\relatedversion{}

\supplement{}

\funding{}

\acknowledgements{}

\ArticleNo{}
\nolinenumbers 
\hideLIPIcs  

\begin{document}

\maketitle

\begin{abstract}
We study the identity problem for matrices, i.e., whether the identity matrix is in a semigroup  generated by a given set of generators.
In particular we consider the identity problem for the \emph{special linear group} following recent
$\NP$-completeness result for \sltwoz and the undecidability for \slfourz generated by $48$ matrices. First we show that there is no embedding from pairs of words into $3\times3$ integer matrices with determinant one, i.e., into {\slthreez} extending previously known result that 
there is no embedding into $\mathbb{C}^{2\times 2}$. Apart from theoretical importance of the result it can be seen as a strong evidence that the computational problems in \slthreez are decidable. The result excludes the most natural possibility of encoding 
the Post correspondence problem into {\slthreez}, where the matrix products extended by the right multiplication correspond to the Turing machine simulation.
Then we show that the \emph{identity problem} is decidable in polynomial time for
an important subgroup of ${\rm SL}(3,\mathbb{Z})$, the Heisenberg group ${\rm H}(3,\mathbb{Z})$.
Furthermore, we extend the decidability result for ${\rm H}(n,\mathbb{Q})$ in any dimension $n$.
Finally we are tightening the gap on decidability question for this long standing open problem
by improving the undecidability result for the identity problem in \slfourz
substantially reducing the bound on the size of the generator set from $48$ to $8$ 
by developing a novel reduction technique.
\end{abstract}

\section{Introduction}
The dynamics of many systems can be represented by matrices and matrix products. The analysis of such systems
lead to solving reachability questions in matrix semigroups which is essential part in verification procedures, control theory questions, biological systems predictability, security etc. \cite{BN16,BM04,COW13,COW16,ding2015,GOW15,KLZ16,MT17,OPW15,OW14,OW14a}.
Many nontrivial algorithms for decision problems on matrix semigroups have been developed for matrices under different constraints on the dimension, the size of a generating set or for specific subclasses of matrices: e.g. commutative matrices \cite{BBC+96}, row-monomial matrices \cite{LP04} or $2 \times 2$ matrix semigroups generated by non-singular integer matrices~\cite{PS17}, upper-triangular integer matrices \cite{Honkala14}, matrices from the special linear group \cite{BHP17,CK05}, etc.

Despite visible interest in this research domain, we still see a significant lack of algorithms and complexity results for answering decision problems in matrix semigroups. Many computational problems for matrix (semi)groups are computationally hard starting from dimension two and very often become undecidable from dimensions three or four even in the case of integer matrices. The central decision problem in matrix semigroups is the membership problem, which was originally considered by A. Markov in 1947 \cite{Markov47}. Let $S = \langle G \rangle$ be a matrix semigroup finitely generated by a generating set of square matrices~$G$. The {\em membership problem} is to decide whether or not a given matrix $M$ belongs to the matrix semigroup $S$. By restricting~$M$ to be the identity matrix we call the problem the {\em identity problem}. 

\begin{problem}[Identity problem]
Let $S = \langle G \rangle$, where $G$ is a finite set of $n$-dimensional matrices over $\mathbb{K}=\Z,\mathbb{Q},\mathbb{R},\mathbb{C},\ldots$. Is the identity matrix in the semigroup, i.e., does $\bm{I}\in S$ hold?
\end{problem}

The identity problem is computationally equivalent to another fundamental problem -- the \emph{subgroup problem} (i.e., to decide whether a semigroup contains a subgroup)
as any subset of matrices, which can form a product leading to the identity also generate a group \cite{CK05}
\footnote{The product of matrices which is equal to the identity is still the identity element after a cyclic shift, so
every element from this product has the inverse.}.

The decidability status of the identity problem was unknown for a long time for matrix semigroups of any dimension, see Problem 10.3 in ``Unsolved Problems in Mathematical Systems and Control Theory'' \cite{BM04}, but it was shown in \cite{BP10} to be undecidable for $48$ matrices from $\mathbb{Z}^{4 \times 4}$ by proving that the identity correspondence problem (a variant of the Post correspondence problem over a group alphabet) is undecidable, and embedding pairs of words over free group alphabet into \slfourz
as two blocks on the main diagonal and by a morphism $f$ as follows
$f(a)=\begin{psmallmatrix}1&2\\0&1\end{psmallmatrix}$,
$f(a^{-1})=\begin{psmallmatrix}1&-2\\0&1\end{psmallmatrix}$, 
$f(b)=\begin{psmallmatrix}1&0\\2&1\end{psmallmatrix}$ and
$f(b^{-1})=\begin{psmallmatrix}1&0\\-2&1\end{psmallmatrix}$.
In the seminal paper of Paterson in 1970, see \cite{Paterson70}, an injective morphism from pairs of words in alphabet $\Sigma=\{a,b\}$ into $3\times3$ integral matrices,
$g(u,v)=\begin{psmallmatrix}n^{|u|}&0&0\\0&n^{|v|}&0\\\sigma(u)&\sigma(v)&1\end{psmallmatrix}$ (where $\sigma$ represents each word as an $n$-adic number),
was used to prove undecidability of mortality and which later led to many undecidability results of matrix problems in dimension three, e.g., \cite{CHK99,HHH07}. Finding new injective morphisms is hard, but having them gives an opportunity to prove new undecidability results.

In 1999, Cassaigne, Harju and Karhum\"aki significantly boosted the research on finding algorithmic solutions for $2 \times 2$ matrix semigroups by showing that there is no injective semigroup morphism from pairs of words over any finite alphabet (with at least two elements) into complex $2 \times 2$ matrices \cite{CHK99}. This result led to substantial interest in finding algorithmic solutions for such problems as the identity problem, mortality, membership, vector reachability, freeness etc. for $2 \times 2$ matrices.

For example, in 2007 Gurevich and Schupp~\cite{GS07} showed that the membership problem is decidable in polynomial time for the finitely generated subgroups of the modular group and later in 2017 Bell, Hirvensalo and Potapov proved that the identity problem for a semigroup generated by matrices from \sltwoz is $\NP$-complete by developing a new effective technique to operate with compressed word representations of matrices and closing the gap on complexity improving the original $\EXPSPACE$ solution proposed in 2005 \cite{CK05}. The first algorithm for the membership problem which covers the cases beyond \sltwoz and \gltwoz has been proposed in \cite{PS17} and provides the solution for a semigroup generated by non-singular $2 \times 2$ integer matrices. Later, these techniques have been applied to build another algorithm to solve the membership problem in \gltwoz extended by singular matrices~\cite{PS17a}. The current limit of decidability is standing for $2 \times 2$ matrices which are defined over hypercomplex numbers (quaternions) for which most of the problems have been shown to be undecidable in \cite{BP08} and correspond to reachability problems for 3-sphere rotation.

In our paper, we show that there is no embedding from pairs of words into $3\times3$ integer matrices
with determinant one (i.e., into \slthreez), 
which is a strong evidence that computational problems in
\slthreez are decidable as all known undecidability techniques for low-dimensional matrices
are based on encoding of Turing machine computations via the Post correspondence problem (\PCP)
which cannot be applied in \slthreez following our result.
In case of the \PCP encoding the matrix products extended by the right multiplication correspond to the Turing
machine simulation and the only known proof alternatives are recursively enumerable sets and
Hilbert's tenth problem that provide undecidability for matrix equations, but of very high dimensions \cite{BHH+08,CH14,Honkala15}. 

So in analogy to 1999 result from \cite{CHK99} on non-existence of embedding into $2 \times 2$ matrix semigroups over complex numbers, we expand a horizon of decidability area for matrix semigroups and show that there is no embedding from a set of pairs of words over a semigroup alphabet to any matrix semigroup in \slthreez. It follows almost immediately that there is no embedding from a set of pairs of group words into $\mathbb{Z}^{3 \times 3}$.\footnote{The idea that such result may hold was motivated by analogy from combinatorial topology, where the identity problem is decidable for the braid group $B_3$ which is the universal central extension of the modular group \psltwoz~\cite{Potapov13}, an embedding for a set of pairs of words into the braid group $B_5$ exists, see \cite{BD99}, and non-existence of embeddings were proved for $B_4$ in \cite{Akimenkov91}. So \slthreez was somewhere in the goldilocks zone between $B_3$ and $B_5$.} The matrix semigroup in \slthreez has attracted a lot of attention recently as it can be represented by a set of generators and relations~\cite{CRW92,Conder16} similar to \sltwoz where it was possible to convert numerical problems into symbolic problems and solve them with novel computational techniques; see \cite{BHP17,CK05,PS17,PS17a}. Comparing to the relatively simple representation of $\sltwoz=\langle S,T \mid S^4= \bm{I}_2, (ST)^6=\bm{I}_2 \rangle$,
where $S=\begin{psmallmatrix}0 & -1 \\ 1 & 0\end{psmallmatrix}$ and $T=\begin{psmallmatrix}1 & 1 \\ 0 & 1\end{psmallmatrix}$
the case of $\slthreez =\langle X,Y,Z \mid X^3=Y^3=Z^2=(XZ)^3=(YZ)^3=(X^{-1}ZXY)^2=(Y^{-1}ZYX)^2=(XY)^6=\bm{I}_3 \rangle$,
where 
\begin{align*}
X=\begin{psmallmatrix}0 & 1 & 0 \\ 0 & 0 & 1 \\ 1 & 0 & 0\end{psmallmatrix},\ Y=\begin{psmallmatrix}1 & 0 & 1 \\ 0 & -1 & -1 \\ 0 & 1 & 0\end{psmallmatrix} \text{ and }Z=
\begin{psmallmatrix}0 & 1 & 0 \\ 1 & 0 & 0 \\ -1 & -1 & -1\end{psmallmatrix},
\end{align*}
looks more challenging containing both non-commutative and partially commutative elements.

As the decidability status of the \emph{identity problem} in dimension three is still a long standing open problem, we look for an important subgroup of ${\rm SL}(3,\mathbb{Z})$, the Heisenberg group ${\rm H}(3,\mathbb{Z})$, for which the \emph{identity problem} could be decidable following our result on non-existence of embedding.
The Heisenberg group is an important subgroup of \slthreez which is useful in the description of one-dimensional quantum mechanical systems~\cite{Brylinski93,GU14,Kostant70}.
We show that the \emph{identity problem} for a matrix semigroup generated by matrices from ${\rm H}(3,\mathbb{Z})$ and even ${\rm H}(3,\mathbb{Q})$ is decidable in polynomial time. Furthermore, we extend the decidability result for ${\rm H}(n,\mathbb{Q})$ in any dimension $n$.

Moreover we tighten the gap between (un)decidability results
for the \emph{identity problem} substantially reducing the bound on the size of
the generator set from $48$ (see \cite{BP10}) to $8$ in \slfourz by developing a novel reduction technique.

\section{Preliminaries}

A {\em semigroup} is a set equipped with an associative binary operation. Let $S$ be a semigroup and $G$ be a subset of $S$. We say that a semigroup $S$ is {\em generated} by a subset~$G$ of~$S$ if each element of $S$ can be expressed as a composition of elements of $G$. In this case, we call $G$ the {\em generating set} of $S$. Given an \emph{alphabet} $\Sigma = \{a_1,a_2, \ldots, a_m\}$, a finite \emph{word} $u$ is an element of semigroup $\Sigma^*$. The \emph{empty word} is denoted by $\varepsilon$. The length of a finite word~$u$ is denoted by $|u|$ and $|\varepsilon|=0$. 

Let $\Gamma=\{a_1,a_2,\ldots,a_\ell,a^{-1}_1,a^{-1}_2,\ldots,a^{-1}_\ell\}$ be a generating set of a free group $\FG$. The elements of $\FG$ are all \emph{reduced} words over $\Gamma$, i.e., words not containing $a_ia_i^{-1}$ or $a_i^{-1}a_i$ as a subword. In this context, we call $\Gamma$ a finite \emph{group alphabet}, i.e., an alphabet with an involution. The multiplication of two elements (reduced words) $u,v\in \FG$ corresponds to the unique reduced word of the concatenation $uv$. This multiplication is called \emph{concatenation} throughout the paper. Later in the encoding of words over a group alphabet we denote $a^{-1}$ by $\overline{a}$ and the alphabet of inverse letters is denoted as $\Sigma^{-1}=\{a^{-1}\mid a\in \Sigma\}$. 

In the next lemma, we present an encoding from an arbitrary group alphabet to a binary group alphabet used in Section~\ref{sec:IPfour}. The result is crucial as it allows us to present the results of the above section over the smallest domain.

\begin{lemma}[Birget, Margolis \cite{BM08}]\label{lem:groupEnc} Let $\Gamma = \{z_1,\ldots, z_\ell,\overbar{z_1},\ldots,\overbar{z_\ell}\}$ be a group alphabet and $\Gamma_2=\{c,d,\overbar{c},\overbar{d}\}$ be a binary 
group alphabet. Define the mapping $\alpha:\Gamma \to \FG[\Gamma_2]$ by:
\begin{align*}
\alpha (z_i) = c^i d\overbar{c}^{i}, \qquad \alpha (\overbar{z_i}) = c^i\overbar{d}\overbar{c}^{i},
\end{align*}
where $1 \leq i \leq \ell$. Then $\alpha$ is a monomorphism, that is, an injective morphism. Note that $\alpha$ can be extended to domain $\FG$ in the usual way.
\end{lemma}

The \emph{Post correspondence problem} (\PCP) is a famous undecidable problem. 
In Section~\ref{sec:IPfour}, we use the \PCP to reduce the number of generators needed to prove the undecidability of the identity problem for \slfourz.
An \emph{instance} of the \PCP consists of two morphisms $g,h: \Sigma^*\rightarrow B^*$, where $\Sigma$ and $B$ are alphabets. A nonempty word $u\in \Sigma^*$ is a \emph{solution} of an instance $(g,h)$ if it satisfies $g(u)=h(u)$. The problem is undecidable for all domain alphabets~$\Sigma$ with $|\Sigma|\geq 5$ \cite{Neary15}.

The special linear group is ${\rm SL}(n,\mathbb{K})=\{M\in \mathbb{K}^{n\times n}\mid \det(M)=1\}$, where $\mathbb{K}=\Z,\mathbb{Q},\mathbb{R},\mathbb{C},\ldots$. The \emph{identity matrix} is denoted by $\bm{I}_n$ and the \emph{zero matrix} is denoted by $\bm{0}_n$. The \emph{Heisenberg group} \heisk is formed by the $3 \times 3$ matrices of the form 
$M  = 
\begin{psmallmatrix}
1 & a & c\\
0 & 1 & b\\
0 & 0 & 1
\end{psmallmatrix}$, 
where $a,b,c \in \mathbb{K}$.
It is easy to see that the Heisenberg group is a non-commutative subgroup of ${\rm SL}(3,\mathbb{K})$. 
We can consider the Heisenberg group as a set of all triples with the following 
group law: 
\begin{align*}
(a_1, b_1, c_1)\otimes (a_2, b_2, c_2) = (a_1 + a_2, b_1 + b_2, c_1 + c_2 + a_1 b_2).
\end{align*}
By $\psi(M)$  we denote the triple~$(a,b,c) \in \mathbb{K}^3$ which corresponds to the upper-triangular coordinates of $M$.
Let $M$ be a matrix in \heisk such that $\psi(M) = (a,b,c)$. We define 
the {\em superdiagonal vector} of $M$ to be $\vec{v}(M) = (a,b)$. 
Given two vectors~$\bu = (u_1,u_2)$ and $\bv = (v_1,v_2)$, the {\em cross 
product} of $\bu$ and $\bv$ is defined as $\bu \times \bv = u_1v_2 - u_2v_1$. 
Two vectors are {\em parallel} if the cross product is zero. 

The Heisenberg group can be also defined in higher dimensions. The Heisenberg group 
of dimension $n$ over $\mathbb{K}$ is denoted by ${\rm H}(n,\mathbb{K})$ and is the group of square matrices in 
$\mathbb{K}^{n \times n}$ of the following form:
$\begin{psmallmatrix}
1 & \ba^\mathsf{T} & c\\
0 & {\bm{I}}_{n-2} & \bb\\
0 & 0 & 1
\end{psmallmatrix}$, 
where $\ba,\bb \in \mathbb{K}^{n-2}, c \in \mathbb{K}$.

Similar to the Heisenberg group in dimension three, we can also consider the Heisenberg group 
in dimension $n$ for any integer $n \ge 3$ as a set of all triples with the following 
group law:
$(\ba_1, \bb_1, c_1)\otimes (\ba_2, \bb_2, c_2) = (\ba_1 + \ba_2, \bb_1 + \bb_2, c_1 + c_2 + \ba_1 \cdot \bb_2)$,
where $\ba_1, \ba_2, \bb_1, \bb_2 \in \mathbb{K}^{n-2}$ and $\ba_1\cdot\bb_2$ is the dot product of vectors $\ba_1$ and $\bb_2$.

We extend the function $\psi$ to $n$-dimensional Heisenberg group: For a matrix $M$, $\psi(M)$ is the triple~$(\ba,\bb,c) \in (\mathbb{K}^{n-2})^2 \times \mathbb{K}$ which corresponds to the upper-triangular coordinates of $M$.

Next, we prove a simple necessary and sufficient condition for commutation of two matrices from the Heisenberg group.

\begin{lemma}\label{lem:commute}
Let $M_1$ and $M_2$ be two matrices from the Heisenberg group~${\rm H}(n,\mathbb{K})$ and $\psi(M_i)=(\ba_i,\bb_i,c_i)$ for $i=1,2$. Then $M_1M_2 = M_2M_1$ holds if and only if
$\ba_1\cdot\bb_2=\ba_2\cdot\bb_1$.\footnote{Note that, in dimension three, the condition can be stated as superdiagonal vectors of $M_1$ and $M_2$ being parallel.}
\end{lemma}
\begin{proof}
The product $M_1M_2$
has $c_1+ c_2 + \ba_1\cdot\bb_2$ in the upper-right corner
whereas $M_2M_1$ has $c_1+ c_2 + \ba_2\cdot\bb_1$. The other
coordinates are identical as we add numbers
in the same coordinate. It is easy to see that the two products are
equivalent if and only if $\ba_1\cdot\bb_2 = \ba_2\cdot\bb_1$ holds.
\end{proof}

The main results of the paper, Theorem~\ref{thm:ptime} and Theorem~\ref{thm:heisnq}, reduce to solving systems of linear homogeneous Diophantine equations. In the next lemma, we show that solving a system of linear homogeneous Diophantine equations is in polynomial time. Note that the polynomial time complexity is not important for Theorem~\ref{thm:heisnq} as there is a part of the decision procedure that requires exponential time.

\begin{lemma}\label{lem:DiophantineP}
Deciding whether a system of linear homogeneous Diophantine equations 
\begin{align}\label{eq:slhDe}
\begin{psmallmatrix}a_{11} & \cdots & a_{1n} \\ \vdots & \ddots & \vdots \\ a_{m1} & \cdots & a_{mn} \end{psmallmatrix}\begin{psmallmatrix}
x_1\\ \vdots \\ x_n
\end{psmallmatrix} &= \begin{psmallmatrix}
0 \\ \vdots \\ 0
\end{psmallmatrix},
\end{align}
where $a_{ij}\in\Q$, has a positive integer solution is in $\P$.
\end{lemma}
\begin{proof}
We prove the claim by converting the system of linear homogeneous Diophantine equations into an instance of linear programming problem which is known to be solvable in polynomial time \cite{Khachiyan80}. Indeed, let us convert \eqref{eq:slhDe} into an instance of the linear programming problem
\begin{align}\label{eq:lp}
\begin{psmallmatrix}a_{11} & \cdots & a_{1n} \\ \vdots & \ddots & \vdots \\ a_{m1} & \cdots & a_{mn} \\ -a_{11} & \cdots & -a_{1n} \\ \vdots & \ddots & \vdots \\ -a_{m1} & \cdots & -a_{mn} \\ -1 & \cdots & -1 \end{psmallmatrix}\begin{psmallmatrix}
y_1\\ \vdots \\ y_n
\end{psmallmatrix} &\leq \begin{psmallmatrix}
0 \\ \vdots \\ 0 \\ 0 \\ \vdots \\ 0 \\ -1
\end{psmallmatrix}.
\end{align}
The idea is that equations 
\begin{align*}
(a_{i1},\ldots,a_{in})\cdot(y_1,\ldots,y_n)^\mathsf{T} &\leq 0 \\
(-a_{i1},\ldots,-a_{in})\cdot(y_1,\ldots,y_n)^\mathsf{T} &\leq 0 
\end{align*} 
ensure that if $(y_1,\ldots,y_n)$ satisfies both equations, it in fact satisfies 
\begin{align*}
(a_{i1},\ldots,a_{in})\cdot(y_1,\ldots,y_n)^\mathsf{T} = 0.
\end{align*}
The final equation guarantees that a solution is not a zero vector.
Let $(\xi_1,\ldots,\xi_n)\in\Q^n$ be a solution of \eqref{eq:lp}. We write $\xi_i=\frac{p_i}{q_i}$ as an irreducible fraction, where $p_i,q_i\in\N$. Now $(\xi'_1,\ldots,\xi'_n)$, where $\xi'_i=\prod_{j=1}^n q_j \xi_i$, is a solution to the system of linear homogeneous Diophantine equations. First, observe that $(\xi'_1,\ldots,\xi'_n)$ is in $\Z^n$ and satisfies the matrix equation. Indeed,
\begin{align*}
\begin{psmallmatrix}a_{11} & \cdots & a_{1n} \\ \vdots & \ddots & \vdots \\ a_{m1} & \cdots & a_{mn} \end{psmallmatrix}\begin{psmallmatrix}
\xi'_1\\ \vdots \\ \xi'_n
\end{psmallmatrix} =
\prod_{j=1}^n q_j \begin{psmallmatrix}a_{11} & \cdots & a_{1n} \\ \vdots & \ddots & \vdots \\ a_{m1} & \cdots & a_{mn} \end{psmallmatrix}\begin{psmallmatrix}
\xi_1\\ \vdots \\ \xi_n
\end{psmallmatrix} = \prod_{j=1}^n q_j \begin{psmallmatrix}
0 \\ \vdots \\ 0
\end{psmallmatrix} 
=\begin{psmallmatrix}
0 \\ \vdots \\ 0
\end{psmallmatrix}.
\end{align*}
Also $(\xi'_1,\ldots,\xi'_n)$ is not a trivial solution (i.e., $(0,\ldots,0)$). Indeed, assume that $\xi'_i=0$ for all $i$. Now, since $\xi'_i=\prod_{j=1}^n q_j \xi_i$ and all $q_j$ are non-zero, $\xi_i=0$. That is, $(\xi_1,\ldots,\xi_n)=(0,\ldots,0)$ which does not satisfy the last equation, i.e., $-1\cdot0-\ldots-1\cdot0=0\leq -1$ does not hold.
Finally, note that $\xi'_i\geq0$ for all $i$. That is, $(\xi'_1,\ldots,\xi'_n)$ is a non-trivial integer solution to the system of linear homogeneous Diophantine equations \eqref{eq:slhDe}. 
\end{proof}

\section{On embedding from pairs of words into \slthreek}
Let $\Sigma=\{0,1\}$. The monoid~$\Sigma^* \times \Sigma^*$ has a generating set $S=\{(0, \varepsilon), (1, \varepsilon), (\varepsilon, 0), (\varepsilon, 1)\}$, where $\varepsilon$ is the empty word. We simplify the notation by setting $a = (0, \varepsilon)$, $b = 
(1, \varepsilon)$, $c = (\varepsilon, 0)$ and $d  = (\varepsilon, 1)$. 
It is easy to see that we have the following relations:
\begin{align}\label{relation}
ac&=ca,           &  bc &=cb,              &  ad &= da, & bd &= db.
\end{align}
In other words, $a$ and $b$ commute with $c$ and $d$. Furthermore, these are the only relations. That is, $a$ and $b$ do not commute with each other, and neither do $c$ and $d$. The monoid $\Sigma^*\times\Sigma^*$ is a partially commutative monoid or a trace monoid. A necessary and sufficient conditions for existence of an embedding of trace monoids into $\N^{2\times2}$ was given in \cite{Choffrut90} but, to the authors' best knowledge, there are no similar results even for $\N^{3\times3}$. Let $\varphi:\Sigma^*\times\Sigma^*\to \slthreek$ be an injective morphism and denote $A = \varphi(a)$, $B=\varphi(b)$, $C = \varphi(c)$ and $D = \varphi(d)$. Our goal is to show that $\varphi$ does not exist for $\mathbb{K}=\Z$. Additionally, we also show that an embedding does exist for $\mathbb{K}=\mathbb{Q}$.
Unfortunately, the technique  developed in \cite{CHK99}, where the contradiction was derived from simple relations, resulting from matrix multiplication, cannot be used for a case of~\slthreek
as it creates a large number of equations which do not directly limit the existence of $\varphi$. 
In contrast to \cite{CHK99}, we found new techniques to show non-existence of $\varphi$ by analysis of eigenvalues and the Jordan normal forms.

We consider Jordan normal forms of matrices and showing that some normal form result in additional relations beside relations \eqref{relation}.

Let $\varphi$ be an injective morphism from $S$ into \slthreec.	Because of obvious symmetries, it suffices to prove the claim for $A=\varphi((0,\varepsilon))$. Now, the only relations in \slthreec are $AC=CA$, $AD=DA$, $BC=CB$ and $BD=DB$.

Since the conjugation by an invertible matrix does not influence the injectivity, we can conjugate the four matrices by some $X\in\mathbb{C}^{3\times3}$ such that $A$ is in the Jordan normal form. For a $3 \times 3$ matrix, there are six different types of matrices in the Jordan normal form. If $A$ has three different eigenvalues, then
\begin{align}\label{eq:threeCigen}
A  = 
\begin{pmatrix}
\lambda & 0 & 0\\
0 & \mu & 0\\
0 & 0 & \nu
\end{pmatrix}.
\end{align}
If $A$ has two eigenvalues, then
\begin{align}\label{eq:twoCigen}
A  = 
\begin{pmatrix}
\lambda & 0 & 0\\
0 & \mu & 0\\
0 & 0 & \mu
\end{pmatrix}
\mbox{  or  }
A  = 
\begin{pmatrix}
\lambda & 0 & 0\\
0 & \mu & 1\\
0 & 0 & \mu
\end{pmatrix}.
\end{align}
Finally, if $A$ has only one eigenvalue, then 
\begin{align}\label{eq:oneCigen}
A  = 
\begin{pmatrix}
\lambda & 0 & 0\\
0 & \lambda & 0\\
0 & 0 & \lambda
\end{pmatrix}
\mbox{  or  }
A  = 
\begin{pmatrix}
\lambda & 1 & 0\\
0 & \lambda & 0\\
0 & 0 & \lambda
\end{pmatrix}
\mbox{  or  }
A  = 
\begin{pmatrix}
\lambda & 1 & 0\\
0 & \lambda & 1\\
0 & 0 & \lambda
\end{pmatrix}.
\end{align}

\begin{lemma}\label{lem:3eval}
Let $\Sigma=\{0,1\}$. If there is an injective morphism $\varphi: \Sigma^* \times \Sigma^* \to {\rm SL}(3,\mathbb{C})$ and the matrices~$A,B,C$ and $D$ correspond to 
$\varphi((0,\varepsilon))$, $\varphi((1,\varepsilon))$, $\varphi((\varepsilon,0))$ and $\varphi((\varepsilon,1))$ respectively, then the Jordan normal form of matrices $A,B,C$ and $D$ is not
$
\begin{psmallmatrix}
\lambda & 0 & 0\\
0 & \mu & 0\\
0 & 0 & \nu
\end{psmallmatrix}$.
\end{lemma}
\begin{proof}
This form can be easily ruled out since $A=\begin{psmallmatrix}
\lambda & 0 & 0\\
0 & \mu & 0\\
0 & 0 & \nu
\end{psmallmatrix}$ only commutes with diagonal matrices. Then $C$ and $D$ should be commuting with $A$ by the suggested relations and as a result, $C$ and $D$ commute with each other.
\end{proof}

\begin{lemma}\label{lem:2evalNotDiag}
Let $\Sigma=\{0,1\}$. If there is an injective morphism $\varphi: \Sigma^* \times \Sigma^* \to {\rm SL}(3,\mathbb{C})$ and the matrices~$A,B,C$ and $D$ correspond to 
$\varphi((0,\varepsilon))$, $\varphi((1,\varepsilon))$, $\varphi((\varepsilon,0))$ and $\varphi((\varepsilon,1))$ respectively, then the Jordan normal form of matrices $A,B,C$ and $D$ is not
$
\begin{psmallmatrix}
\lambda & 0 & 0\\
0 & \mu & 1\\
0 & 0 & \mu
\end{psmallmatrix}$.
\end{lemma}
\begin{proof}
Let $A=\begin{psmallmatrix}\lambda&0&0\\0&\mu&1\\0&0&\mu\end{psmallmatrix}$
and let $C=\begin{psmallmatrix}a&b&c \\ d&e&f \\ g&h&\ell\end{psmallmatrix}$. Now
\begin{align*}
AC&=\begin{pmatrix}\lambda&0&0\\0&\mu&1\\0&0&\mu\end{pmatrix}
\begin{pmatrix}a&b&c\\d&e&f\\g&h&\ell\end{pmatrix}
=
\begin{pmatrix}\lambda a&\lambda b&\lambda c\\g + \mu d&h +\mu e& \ell +\mu f\\ \mu g&\mu h & \mu \ell\end{pmatrix} \text{ and } \\
CA&=\begin{pmatrix}a&b&c\\d&e&f\\g&h&\ell\end{pmatrix}\begin{pmatrix}\lambda&0&0\\0&\mu&1\\0&0&\mu\end{pmatrix}=\begin{pmatrix}\lambda a&\mu b&b+\mu c\\\lambda d&\mu e&e+\mu f\\\lambda g&\mu h& h+\mu \ell\end{pmatrix}.
\end{align*}
Since these matrices are equal, and since $\lambda\neq\mu$, we have that $b=c=d=g=h=0$ and $e=\ell$. Similar calculation gives us $D=\begin{psmallmatrix}a'&0&0\\0&e'&f'\\0&0&e'\end{psmallmatrix}$. Now,  matrices $C$ and $D$ commute as follows:
\begin{align*}
\begin{pmatrix}a&0&0\\0&e&f\\0&0&e\end{pmatrix}\begin{pmatrix}a'&0&0\\0&e'&f'\\0&0&e'\end{pmatrix}=\begin{pmatrix}aa'&0&0\\0&ee'&ef'+fe'\\0&0&ee'\end{pmatrix} = \begin{pmatrix}a'&0&0\\0&e'&f'\\0&0&e'\end{pmatrix}\begin{pmatrix}a&0&0\\0&e&f\\0&0&e\end{pmatrix},
\end{align*}
which is not one of the allowed relations. 
\end{proof}

\begin{lemma}\label{lem:1eval}
Let $\Sigma=\{0,1\}$. If there is an injective morphism $\varphi: \Sigma^* \times \Sigma^* \to {\rm SL}(3,\mathbb{C})$ and the matrices~$A,B,C$ and $D$ correspond to 
$\varphi((0,\varepsilon))$, $\varphi((1,\varepsilon))$, $\varphi((\varepsilon,0))$ and $\varphi((\varepsilon,1))$ respectively, then the Jordan normal form of matrices $A,B,C$ and $D$ is not
$
\begin{psmallmatrix}
\lambda & 0 & 0\\
0 & \lambda & 0\\
0 & 0 & \lambda
\end{psmallmatrix}$ nor $
\begin{psmallmatrix}
\lambda & 1 & 0\\
0 & \lambda & 1\\
0 & 0 & \lambda
\end{psmallmatrix}$.
\end{lemma}
\begin{proof}
In the first case, the matrix~$A$ is diagonal and it is easy to see that then $A$ 
commutes with all matrices, including $B$. 

Let us then consider the second case, where the matrix~$A$ is in the following form~$\begin{psmallmatrix}\lambda &1&0 \\ 0&\lambda&1 \\ 0&0&\lambda\end{psmallmatrix}$ and let $C=\begin{psmallmatrix}a&b&c \\ d&e&f \\ g&h&\ell\end{psmallmatrix}$. Now
\begin{align*}
AC&=\begin{pmatrix}\lambda&1&0\\0&\lambda&1\\0&0&\lambda\end{pmatrix}\begin{pmatrix}a&b&c\\d&e&f\\g&h&\ell\end{pmatrix}=\begin{pmatrix}d+a\lambda&e+b\lambda&f+c\lambda\\g + d \lambda&h + e \lambda& \ell + f \lambda\\g \lambda&h \lambda&\ell \lambda\end{pmatrix} \text{ and}\\
CA&=\begin{pmatrix}a&b&c\\d&e&f\\g&h&\ell\end{pmatrix}\begin{pmatrix}\lambda&1&0\\0&\lambda&1\\0&0&\lambda\end{pmatrix}=\begin{pmatrix}a \lambda&a+b \lambda&b+c \lambda\\d \lambda&d+e \lambda&e+f \lambda\\g \lambda&g+h \lambda&h+\ell \lambda\end{pmatrix}.
\end{align*}
Since these matrices are equal, we have that $d=g=h=0$, $a=e=\ell$ and $b=f$. Let $D=\begin{psmallmatrix}a'&b'&c'\\d'&e'&f'\\g'&h'&\ell'\end{psmallmatrix}$. Solving $D$ from equation $AD=DA$, gives us $D=\begin{psmallmatrix}a'&b'&c'\\0&a'&b'\\0&0&a'\end{psmallmatrix}$ and now matrices $C$ and $D$ commute by Lemma~\ref{lem:commute}. Indeed, matrix $C$ can be expressed as $C=a\begin{psmallmatrix}1&\frac{b}{a}&\frac{c}{a}\\0&1&\frac{b}{a}\\0&0&1\end{psmallmatrix}\in \heisc$ and matrix $D$ has an analogous expression. Then it is clear that $\frac{b}{a}\frac{b'}{a'}=\frac{b'}{a'}\frac{b}{a}$ and thus matrices $C$ and $D$ commute.
\end{proof}

In the above lemmas, we ruled out four out of six possible Jordan normal forms. In the next theorem, we give an embedding from $\Sigma\times\Sigma$ into \slthreeq.
\begin{theorem}
Let $\Sigma=\{0,1\}$. The morphism $\varphi: \Sigma^* \times \Sigma^* \to {\rm SL}(3,\mathbb{Q})$, defined by $\varphi((0,\varepsilon))=\begin{psmallmatrix}
4 & 0 & 0\\
0 & \frac{1}{2} & 0\\
0 & 0 & \frac{1}{2}
\end{psmallmatrix}$, $\varphi((1,\varepsilon))=\begin{psmallmatrix}
9 & \frac{1}{3} & 0\\
0 & \frac{1}{3} & 0\\
0 & 0 & \frac{1}{3}
\end{psmallmatrix}$, $\varphi((\varepsilon,0))=\begin{psmallmatrix}
\frac{1}{2} & 0 & 0\\
0 & \frac{1}{2} & 0\\
0 & 0 & 4
\end{psmallmatrix}$ and $\varphi((\varepsilon,1))=\begin{psmallmatrix}
\frac{1}{3} & 0 & 0\\
0 & \frac{1}{3} & 0\\
0 & \frac{1}{3} & 9
\end{psmallmatrix}$ is an embedding.
\end{theorem}
\begin{proof}
Let $A=\varphi((0,\varepsilon))=\begin{psmallmatrix}
4 & 0 & 0\\
0 & \frac{1}{2} & 0\\
0 & 0 & \frac{1}{2}
\end{psmallmatrix}$, $B=\varphi((1,\varepsilon))=\begin{psmallmatrix}
9 & \frac{1}{3} & 0\\
0 & \frac{1}{3} & 0\\
0 & 0 & \frac{1}{3}
\end{psmallmatrix}$, $C=\varphi((\varepsilon,0))=\begin{psmallmatrix}
\frac{1}{2} & 0 & 0\\
0 & \frac{1}{2} & 0\\
0 & 0 & 4
\end{psmallmatrix}$ and $D=\varphi((\varepsilon,1))=\begin{psmallmatrix}
\frac{1}{3} & 0 & 0\\
0 & \frac{1}{3} & 0\\
0 & \frac{1}{3} & 9
\end{psmallmatrix}$. It is easy to see that the relations of \eqref{relation} hold. For example, the relation $AB\neq BA$ holds since
\begin{align*}
AB=\begin{pmatrix}
4 & 0 & 0\\
0 & \frac{1}{2} & 0\\
0 & 0 & \frac{1}{2}
\end{pmatrix}\begin{pmatrix}
9 & \frac{1}{3} & 0\\
0 & \frac{1}{3} & 0\\
0 & 0 & \frac{1}{3}
\end{pmatrix} &= \begin{pmatrix}
36&\frac{4}{3}&0\\0&\frac{1}{6}&0\\0&0&\frac{1}{6}
\end{pmatrix} \\
BA=\begin{pmatrix}
9 & \frac{1}{3} & 0\\
0 & \frac{1}{3} & 0\\
0 & 0 & \frac{1}{3}
\end{pmatrix}\begin{pmatrix}
4 & 0 & 0\\
0 & \frac{1}{2} & 0\\
0 & 0 & \frac{1}{2}
\end{pmatrix} &= \begin{pmatrix}
36 & \frac{1}{6} & 0 \\ 0 & \frac{1}{6} & 0 \\ 0 & 0 & \frac{1}{6}
\end{pmatrix}.
\end{align*}
On the other hand, 
\begin{align*}
AD=\begin{pmatrix}
4 & 0 & 0\\
0 & \frac{1}{2} & 0\\
0 & 0 & \frac{1}{2}
\end{pmatrix}\begin{pmatrix}
\frac{1}{3} & 0 & 0\\
0 & \frac{1}{3} & 0\\
0 & \frac{1}{3} & 9
\end{pmatrix} &= \begin{pmatrix}
\frac{4}{3}&0&0\\0&\frac{1}{6}&0\\0&\frac{1}{6}&\frac{9}{2}
\end{pmatrix} = 
\begin{pmatrix}
\frac{1}{3} &0& 0\\
0 & \frac{1}{3} & 0\\
0 &  \frac{1}{3} & 9
\end{pmatrix}\begin{pmatrix}
4 & 0 & 0\\
0 & \frac{1}{2} & 0\\
0 & 0 & \frac{1}{2}
\end{pmatrix}=DA.
\end{align*}

Next, we show that the pairs $\{A,B\}$ and $\{C,D\}$ generate free semigroups. Denote by $A'=\begin{psmallmatrix}4&0\\0&\frac{1}{2}\end{psmallmatrix}$ and $B'=\begin{psmallmatrix}9&\frac{1}{3}\\0&\frac{1}{3}\end{psmallmatrix}$ the top left 2-by-2 blocks of $A$ and $B$ respectively. By Lemma 3 of \cite{CHK99}, $\{A',B'\}^+$ is free if and only if $\{\lambda A',\mu B'\}$ for some $\lambda,\mu\in \Q\setminus \{0\}$. Let $\lambda=2$ and $\mu=3$ and denote $A''=\lambda A'=\begin{psmallmatrix}8&0\\0&1\end{psmallmatrix}$ and $B''=\mu B'=\begin{psmallmatrix}27&1\\0&1\end{psmallmatrix}$. Further, by Proposition 3 of \cite{CHK99}, if $\frac{1}{|a|}+\frac{1}{|b|}\leq1$, then $\{A'',B''\}^+$ is free, where $a$ and $b$ are the elements in top left corner of $A''$ and $B''$ respectively. In our case, the condition holds as $\frac{1}{8}+\frac{1}{27}=\frac{35}{216}<1$ and hence $\{A'',B''\}^+$ and $\{A',B'\}^+$ are free groups. It is easy to see that also $\{A,B\}^+$ is a free group. Proving that $C$ and $D$ generate a free semigroup is done analogously. 

As the matrices $A$, $B$, $C$ and $D$ satisfy the relations of \eqref{relation} and pairwise generate free groups, we conclude that $\varphi$ is an embedding of $\Sigma^*\times\Sigma^*$.
\end{proof}

Note that in the previous lemma, all matrices have a Jordan normal form of $\begin{psmallmatrix}\lambda&0&0\\0&\mu&0\\0&0&\mu\end{psmallmatrix}$, where $\lambda\neq\mu$.
Next, we consider existence of an embedding into \slthreez. Lemmas~\ref{lem:3eval}, \ref{lem:2evalNotDiag} and \ref{lem:1eval} can be applied to matrices of \slthreez with a caveat that the matrices are no longer in \slthreez after $A$ is transformed into a Jordan normal form. In these cases, it does not lead to problems as the contradictions derived in \slthreec do not rely on properties of complex numbers, so they are applicable to integral matrices as well. In \slthreez, we can prove that an embedding, where the Jordan normal form of matrices is $\begin{psmallmatrix}\lambda&0&0\\0&\mu&0\\0&0&\mu\end{psmallmatrix}$, does not exist.

\begin{lemma}\label{lem:another2eval}
Let $\Sigma=\{0,1\}$. If there is an injective morphism $\varphi: \Sigma^* \times \Sigma^* \to {\rm SL}(3,\mathbb{Z})$ and the matrices~$A,B,C$ and $D$ correspond to 
$\varphi((0,\varepsilon))$, $\varphi((1,\varepsilon))$, $\varphi((\varepsilon,0))$ and $\varphi((\varepsilon,1))$ respectively, then the Jordan normal form of matrices $A,B,C$ and $D$ is not
$
\begin{psmallmatrix}
\lambda & 0 & 0\\
0 & \mu & 0\\
0 & 0 & \mu
\end{psmallmatrix}$.
\end{lemma}
\begin{proof}
As in the previous proofs, we assume that $A=\begin{psmallmatrix}
\lambda & 0 &0 \\ 0&\mu&0 \\ 0&0&\mu
\end{psmallmatrix}$. Note that $A,B,C,D\in\slthreec$. Observe that $\det(A)=\lambda\mu^2=1$ and $\tr(\varphi(0,\varepsilon))=\tr(A)=\lambda+2\mu\in\Z$. We claim that $\lambda\in\Z$ and $\mu\in\Q$. Let $A_1=\varphi(0,\varepsilon)$. First, we can rule out that the eigenvalues are complex. Indeed, $\lambda$ and $\mu$ are the solutions to the characteristic polynomial 
\begin{align*}
x^3+\tr(A_1)x^2-\frac{\tr(A_1)^2-\tr(A_1^2)}{2}x+\frac{\tr(A_1)^3+2\tr(A_1^3)-3\tr(A_1)\tr(A_1^2)}{6}
\end{align*}
with rational coefficients. It is well-known that a cubic equation with real coefficients has a complex roots only if there are three distinct roots. In our case, $\mu$ is a double root, and hence both $\lambda$ and $\mu$ are real. Next we show that $\lambda$ and $\mu$ are not irrational. From the general solution for a cubic equation, it follows that \begin{align*}
\lambda &=\frac{2\tr(A_1)(\tr(A_1)^2-\tr(A_1^2))-9\frac{\tr(A_1)^3+2\tr(A_1^3)-3\tr(A_1)\tr(A_1^2)}{6}-\tr(A_1)^3}{-\tr(A_1)^2+3\frac{\tr(A_1)^2-\tr(A_1^2)}{2}} \text{ and} \\
\mu &= \frac{-9\frac{\tr(A_1)^3+2\tr(A_1^3)-3\tr(A_1)\tr(A_1^2)}{6}+\tr(A_1)\frac{\tr(A_1)^2-\tr(A_1^2)}{2}}{2\tr(A_1)^2-3(\tr(A_1)^2-\tr(A_1^2))}.
\end{align*} 
It is clear that both eigenvalues are rational. Finally, we prove that $\lambda$ is in fact integer. Assume to the contrary that $\lambda$ is a rational number represented by irreducible fraction $\frac{n}{m}$. Since $\det(A)=\lambda\mu^2=1$, it follows that $\mu=\frac{\sqrt{m}}{\sqrt{n}}$. Now $\lambda+2\mu=\frac{n}{m}+2\frac{\sqrt{m}}{\sqrt{n}}\in\Z$ only if the denominators are equal, that is, $\sqrt{n}=m$. Then $m^2=n$ and $\lambda=\frac{n}{m}=\frac{m^2}{m}=m\in\Z$ which contradicts our assumption that $\lambda$ is a rational number. Hence, $\lambda$ is an integer. From $\lambda+2\mu\in\Z$ it follows that $2\mu\in \Z$. Denote $\mu=\frac{k}{2}$ for some $k\in\Z$. Now, $\lambda \frac{k^2}{4}=1$ and thus $\lambda k^2=4$. The only integer solutions for $\lambda$ are 1 or 4. Clearly $\lambda\neq1$ as then $\mu$ is also 1 which contradicts our assumption that $\lambda$ and $\mu$ are distinct.

That is, we have concluded that $\lambda=4$ and $\mu=\frac{1}{2}$. Consider then the trace of $A^2$ which is also an integer. Indeed, $\tr(A^2)=\tr(A_1^2)\in\Z$. That is $\lambda^2+2\mu^2=16+2\frac{1}{4}\in\Z$, which is a contradiction.
\end{proof}

With the previous lemmas, we have ruled out the possible Jordan normal forms of potential embeddings into \slthreez. The final Jordan normal form is ruled out in the next lemma.

\begin{lemma}\label{lem:finalform}
Let $\Sigma=\{0,1\}$. If there is an injective morphism $\varphi: \Sigma^* \times \Sigma^* \to {\rm SL}(3,\mathbb{Z})$ and the matrices~$A,B,C$ and $D$ correspond to 
$\varphi((0,\varepsilon))$, $\varphi((1,\varepsilon))$, $\varphi((\varepsilon,0))$ and $\varphi((\varepsilon,1))$ respectively, then the Jordan normal form of matrices $A,B,C$ and $D$ is not
$
\begin{psmallmatrix}
\lambda & 1 & 0\\
0 & \lambda & 0\\
0 & 0 & \lambda
\end{psmallmatrix}$.
\end{lemma}
\begin{proof}
Assume to the contrary that exists an injective morphism $\varphi$ from $\Sigma^* \times \Sigma^*$ into \slthreez.	Since the conjugation by an invertible matrix does not influence the injectivity, we suppose that the image of $a$ is in the Jordan normal form $\begin{psmallmatrix}
\lambda & 1 & 0\\
0 & \lambda & 0\\
0 & 0 & \lambda
\end{psmallmatrix}$ as the other form have been ruled out in the previous lemmas. 
By $A$, $B$, $C$ and $D$ we denote the images of the generators, $a$, $b$, $c$ and $d$, conjugated by the matrix transforming $A$ into the Jordan normal form. Then we have the following matrices corresponding to the generators~$a$, $b$, $c$ and $d$ as follows: 
{\small\begin{align*}
\!\!\!\!\!A &=\begin{pmatrix}
\lambda & 1 & 0\\
0 & \lambda & 0\\
0 & 0 & \lambda
\end{pmatrix},   &         B &= \begin{pmatrix}
a_B & b_B & c_B\\
d_B & e_B & f_B\\
g_B & h_B & \ell_B
\end{pmatrix}, &
C&=\begin{pmatrix}
a_C & b_C & c_C\\
d_C & e_C & f_C\\
g_C & h_C & \ell_C
\end{pmatrix},    &        D &=\begin{pmatrix}
a_D & b_D & c_D\\
d_D & e_D & f_D \\
g_D & h_D & \ell_D
\end{pmatrix}.
\end{align*}}
Note again that $B,C,D\in\slthreec$.

Since $A$ and $C$ commute with each other by one of 
the given relations in~\eqref{relation}, 
we have
{\small\begin{align*}
AC
=
\begin{pmatrix}
\lambda a_C + d_C & \lambda b_C + e_C & \lambda c_C + f_C\\
\lambda d_C & \lambda e_C & \lambda f_C\\
\lambda g_C & \lambda h_C & \lambda\ell_C
\end{pmatrix} = 
\begin{pmatrix}
\lambda a_C & a_C + \lambda b_C & \lambda c_C\\
\lambda d_C & d_C + \lambda e_C & \lambda f_C\\
\lambda g_C & g_C + \lambda h_C & \lambda \ell_C
\end{pmatrix} 
= CA.
\end{align*}}
It is easy to see 
that $d_C = g_C = f_C = 0$ and $a_C = e_C$. 
Therefore, 
{\small\begin{align*}
C = 
\begin{pmatrix}
a_C & b_C & c_C\\
0 & a_C & 0\\
0 & h_C &\ell_C
\end{pmatrix} \text{ and }
D = 
\begin{pmatrix}
a_D & b_D & c_D\\
0 & a_D & 0\\
0 & h_D & \ell_D
\end{pmatrix}.
\end{align*}}
Since $\varphi(c)$ and $\varphi(d)$ are in \slthreec, the determinants of $C$ and $D$ 
are 1. Now, the determinant of $C$ is 
$a_C^2 \ell_C$ and the eigenvalues are $a_C$ and $\ell_C$. As $C$ is similar to $\varphi(c)$, the matrices have the same eigenvalues. Now, $a_C = \ell_C$ as other Jordan normal forms have been ruled out previously. Analogously, we can also 
see that $a_D = \ell_D$.
Next, we observe that the matrices $C$ and $D$ commute if and only if $c_C h_D = c_D h_C$. Indeed, 
\begin{align*}
\!\!\!\!\!CD = 
\begin{pmatrix}
 a_C & b_C & c_C\\
0 &  a_C & 0\\
0 & h_C & a_C
\end{pmatrix}
\begin{pmatrix}
 a_D & b_D & c_D\\
0 &  a_D & 0\\
0 & h_D & a_D
\end{pmatrix} 
&= 
\begin{pmatrix}
 a_Ca_D & b_Ca_D + a_Cb_D + c_C h_D & c_Ca_D + a_Cc_D\\
0 &  a_Ca_D & 0\\
0 & h_Ca_D + a_Ch_D &a_Ca_D
\end{pmatrix} \text{ and} \\
DC = 
\begin{pmatrix}
 a_D & b_D & c_D\\
0 &  a_D & 0\\
0 & h_D &a_D
\end{pmatrix}
\begin{pmatrix}
 a_C & b_C & c_C\\
0 &  a_C & 0\\
0 & h_C &a_C
\end{pmatrix} 
&= 
\begin{pmatrix}
 a_Da_C & b_Da_C + a_Db_C + c_D h_C & c_Da_C + a_Dc_C\\
0 &  a_Da_C & 0\\
0 & h_Da_C + a_Dh_C &a_Da_C
\end{pmatrix}.
\end{align*}
By relations~\eqref{relation}, $C$ and $D$ do not commute and hence there are three cases to be considered: 
\begin{enumerate}
\item $c_C = 0$  and $h_C \ne 0$; 
\item $c_C \ne 0$  and  $h_C = 0$;
\item $c_C \ne 0$  and  $h_C \ne 0$. 
\end{enumerate}

We prove that each case leads to a contradiction, i.e., that $C$ and $D$ commute. 
Let us examine the three cases in more details.

First, let us consider the case where $c_C = 0$ and $h_C \ne 0$.
We know that $c_D$ is also non-zero because 
otherwise $C$ and $D$ commute with each other since 
$c_C h_D = c_D h_C = 0$. We have the following calculations:
\begin{align*}
BC &= 
\begin{pmatrix}
a_B & b_B & c_B\\
d_B & e_B & f_B\\
g_B & h_B & \ell_B
\end{pmatrix}
\begin{pmatrix}
 a_C & b_C & 0\\
0 &  a_C & 0\\
0 & h_C &a_C
\end{pmatrix}
=
\begin{pmatrix}
a_Ba_C & a_B b_C + b_B a_C + c_B h_C & c_Ba_C\\
d_Ba_C &  d_B b_C + e_Ba_C + f_B h_C & f_Ba_C\\
g_Ba_C & g_B b_C + h_Ba_C + \ell_B h_C & \ell_Ba_C
\end{pmatrix} \text{ and} \\
CB &= 
\begin{pmatrix}
 a_C & b_C & 0\\
0 &  a_C & 0\\
0 & h_C &a_C
\end{pmatrix}
\begin{pmatrix}
a_B & b_B & c_B\\
d_B & e_B & f_B\\
g_B & h_B & \ell_B
\end{pmatrix}
=
\begin{pmatrix}
a_Ba_C + d_B b_C & b_Ba_C + e_B b_C & c_Ba_C + f_B b_C\\
d_Ba_C &  e_Ba_C & f_Ba_C\\
d_B h_C + g_Ba_C & e_B h_C + h_Ba_C & f_B h_C + \ell_Ba_C
\end{pmatrix}.
\end{align*}

Since $BC = CB$, we have $d_B b_C = 0$, $d_B h_C = 0$, $f_B b_C = 0$, and 
$f_B h_C = 0$. By the supposition $h_C \ne 0$, we further deduce that 
$d_B = f_B = 0$. Then 
$
B = 
\begin{psmallmatrix}
a_B & b_B & c_B\\
0 & e_B & 0\\
g_B & h_B & \ell_B
\end{psmallmatrix}.
$
Note that we also have 
\begin{align}\label{eq:fromBC}
a_Bb_C + c_Bh_C = e_B b_C \mbox{ and }g_Bb_C + \ell_Bh_C = e_Bh_C
\end{align}
by the equality~$BC = CB$.

The characteristic polynomial of $B$ is
\begin{align*}
P(x) = -x^3 + \tr(B) x^2 - (a_Be_B + a_B\ell_B + e_B\ell_B - c_Bg_B)x + \det(B)
\end{align*}
which has roots $\lambda = e_B$ and $\lambda = \frac{1}{2}(a_B + \ell_B \pm \sqrt{(a_B -\ell_B)^2 + 4c_Bg_B})$. By the previous considerations, we know that $B$ has only one eigenvalue and therefore, we have $a_B =e_B = \ell_B$ and $c_Bg_B = 0$.

Moreover, it follows from \eqref{eq:fromBC} that $c_B = 0$ and $g_B b_C = 0$.
Note that $g_B \ne 0$ because otherwise the matrix $B$ commutes with $A$. 
Finally, we consider
{\small\begin{align*}
BD &= 
\begin{pmatrix}
a_B & b_B & 0\\
0 & a_B & 0\\
g_B & h_B & a_B
\end{pmatrix}
\begin{pmatrix}
 a_D & b_D & c_D\\
0 &  a_D & 0\\
0 & h_D &a_D
\end{pmatrix}
=
\begin{pmatrix}
a_Ba_D & b_Ba_D + a_Bb_D  & a_Bc_D\\
0 &  a_Ba_D  & 0\\
g_Ba_D &  g_Bb_D + h_Ba_D + a_Bh_D & a_Ba_D g_Bc_D
\end{pmatrix}\\
DB&=\begin{pmatrix}
 a_D & b_D & c_D\\
0 &  a_D & 0\\
0 & h_D &a_D
\end{pmatrix}
\begin{pmatrix}
a_B & b_B & 0\\
0 & a_B & 0\\
g_B & h_B & a_B
\end{pmatrix}=
\begin{pmatrix}
a_Da_B + c_D g_B  & a_Db_B + b_Da_D + c_D h_B & c_Da_B\\
0 &  a_Da_B & 0\\
g_Ba_D & a_Dh_B + h_Da_B & a_Da_B
\end{pmatrix}.
\end{align*}}
It is easy to see that $g_Bb_D = g_Bc_D = 0$ and thus $b_D=c_D=0$, and then $D$ commutes with $C$. Therefore, we have a contradiction.

Let us consider the second case where $c_C \ne 0$ and $h_C = 0$. It is quite similar to the previous case.
Consider the matrix~$B$ which commutes with $C$ as follows:
\begin{align*}
BC &= \begin{pmatrix}
a_B & b_B & c_B\\
d_B & e_B & f_B\\
g_B & h_B & \ell_B
\end{pmatrix}
\begin{pmatrix}
 a_C & b_C & c_C\\
0 &  a_C & 0\\
0 & 0 &a_C
\end{pmatrix} =
\begin{pmatrix}
a_Ba_C & a_B b_C + b_Ba_C  & a_B c_C + c_Ba_C\\
d_Ba_C &  d_B b_C + e_Ba_C  & d_B c_C + f_Ba_C\\
g_Ba_C & g_B b_C + h_Ba_C  &g_B c_C + \ell_Ba_C
\end{pmatrix}\\
&=
\begin{pmatrix}
a_Ba_C + d_B b_C + g_B c_C & b_Ba_C + e_B b_C + h_B c_C & c_Ba_C + f_B b_C + \ell_B c_C\\
d_Ba_C &  e_Ba_C & f_Ba_C\\
 g_Ba_C &  h_Ba_C &  \ell_Ba_C
\end{pmatrix}\\&
=
\begin{pmatrix}
 a_C & b_C & c_C\\
0 &  a_C & 0\\
0 & 0 &a_C
\end{pmatrix}
\begin{pmatrix}
a_B & b_B & c_B\\
d_B & e_B & f_B\\
g_B & h_B & \ell_B
\end{pmatrix}
=
CB.
\end{align*}
By the equivalence, we have $d_B b_C = 0$, $g_B b_C = 0$, $g_B c_C = 0$, and 
$d_B c_C = 0$. By the supposition $c_C \ne 0$, we further deduce that 
$d_B = g_B = 0$. Then $B$ is of the following form:
$
B = 
\begin{psmallmatrix}
a_B & b_B & c_B\\
0 & e_B & f_B\\
0 & h_B & \ell_B
\end{psmallmatrix}.
$
Note that we also have 
\begin{align}\label{eq:fromBC2}
a_B b_C = e_Bb_C + h_Bc_C \mbox{ and }a_Bc_C = f_Bb_C + \ell_Bc_C
\end{align}
by the equality~$BC = CB$.

The characteristic polynomial of $B$ is
$P(x) = -x^3 + \tr(B) x^2 - (a_Be_B + a_B\ell_B + e_B\ell_B - f_Bh_B)x + \det(B)$ which 
has roots $\lambda = e_B$ and $\lambda = \frac{1}{2}(e_B + \ell_B \pm \sqrt{(a_B -\ell_B)^2 + 4f_Bh_B})$. We know that $B$ has only one eigenvalue by the previous considerations and therefore, we have $a_B =e_B = \ell_B$ and $f_Bh_B = 0$.

We can further deduce from \eqref{eq:fromBC2} that $h_B = 0$ and $f_Bb_C = 0$. By a similar argument for the matrices $B$ and $D$ that should commute with each other as in the first case, we have a contradiction.
 
Finally, consider the third case where $c_C \ne 0$ and $h_C \ne 0$.
It is obvious that 
$c_D$ and $h_D$ are also non-zero because otherwise $C$ and $D$ would 
commute.
Now consider the matrix~$B$ which is commuting with $C$ and $D$.
We can deduce from the relation~$BC = CB$ that 
$d_B = g_B = f_B = 0$ and $a_B = e_B = \ell_B$ since they are 
eigenvalues of $B$. Hence, 
$
B = 
\begin{psmallmatrix}
a_B & b_B & c_B\\
0 & a_B & 0\\
0 & h_B & a_B
\end{psmallmatrix}.
$

Now we have $c_C h_B = c_B h_C$ since $B$ and $C$ commute with each other.
Note that $h_B$ and $c_B$ are both non-zero since $A$ and $B$ commute 
if $h_B = c_B = 0$. Let us denote $\frac{c_C}{h_C} = \frac{c_B}{h_B} = x$. 
We also have $c_D h_B = c_B h_D$ from the relation~$BD = DB$ and 
have $\frac{c_D}{h_D} = \frac{c_B}{h_B} = x$. From $x = \frac{c_C}{h_C} 
= \frac{c_D}{h_D}$, we have $c_C h_D = c_D h_C$ which results in 
the relation~$CD = DC$. Therefore, we also have a contradiction.
\end{proof}

\begin{theorem}
\label{NoIntoSL3Z}
There is no injective morphism
$
\varphi: \Sigma^* \times \Sigma^* \to {\rm SL}(3,\mathbb{Z})
$
for any binary alphabet~$\Sigma$.
\end{theorem}
\begin{proof}
Since we have examined all possible cases in Lemmas~\ref{lem:3eval}, \ref{lem:2evalNotDiag}, \ref{lem:1eval}, \ref{lem:another2eval} and \ref{lem:finalform} and found contradictions for every case, we can conclude that there is no injective morphism from 
$\Sigma^* \times \Sigma^*$ into the special linear group \slthreez.
\end{proof}

\begin{corollary}
There is no injective morphism
$
\varphi: \FG \times \FG \to \mathbb{Z}^{3\times 3}
$ for any binary group alphabet $\Gamma$.
\end{corollary}
\begin{proof}
We proceed by contradiction. Assume that there exists such an injective morphism~$\varphi$ 
from the set of pairs of words over a group alphabet to the set of matrices in 
$\mathbb{Z}^{3\times 3}$. Suppose that $A = \varphi((a,\varepsilon))$, where $a \in \Gamma$. 
Then the inverse matrix $A^{-1}$ corresponding to $(\overbar{a},\varepsilon)$ must be 
in $\mathbb{Z}^{3\times 3}$. This implies that the determinant of $A$ is $\pm1$ because 
otherwise the determinant of $A^{-1}$ becomes a non-integer. Consider then a morphism~$\psi$ such that $\psi(x)=\varphi(x)\varphi(x)$ for each $x\in \FG\times\FG$. It is clear that also $\psi$ is injective and that the determinant of the image is 1. By Theorem~\ref{NoIntoSL3Z}, such injective morphism $\psi$ does not exist even from semigroup alphabets and hence neither does~$\varphi$.
\end{proof}

\section{Decidability of the identity problem in the Heisenberg group}

The decidability of the identity problem in dimension three is a long standing open problem.
Following our findings on non-existence of embedding into \slthreez, 
 in this section we consider the decidability of an important subgroup of \slthreez ,
the Heisenberg group, which is well-known in the context of quantum mechanical systems \cite{Brylinski93,GU14,Kostant70}.
Recently a few decidability results have been obtained for a knapsack variant of the membership
problem in dimension three (i.e., \heis), where the goal was to solve a single matrix equation with a specific order of matrices \cite{KLZ16}.

In this section, we prove that the identity problem is decidable for the Heisenberg group over rational numbers.
First, we provide more intuitive solution for dimension three, i.e., \heisq, which still requires 
a number of techniques to estimate possible values of elements under permutations in matrix products. 
In the end of the section, we generalize the result for \heisnq case using analogies 
in the solution for dimension three.


Here we prove that the identity problem for matrix semigroups in the Heisenberg group over rationals is decidable by analysing the 
behaviour of multiplications especially in the upper-right coordinate 
of matrices. From Lemma~\ref{lem:commute}, it follows that the matrix multiplication is 
commutative in the Heisenberg group if and only if matrices have 
pairwise parallel superdiagonal vectors. So we analyse two cases of products 
for matrices with pairwise parallel and none pairwise parallel superdiagonal vectors
and then provide algorithms that solve the problem in polynomial time.
The most difficult part is showing that only limited number of conditions must be checked to guarantee the existence of a product that results in the identity.


\begin{lemma}\label{lem:single}
Let $G = \{ M_1, M_2, \ldots, M_r \} \subseteq {\rm H}(3,\mathbb{Q})$ be a set of 
matrices from the Heisenberg group such that superdiagonal vectors of matrices are 
pairwise parallel. If there exists a sequence of matrices
$M = M_{i_1} M_{i_2} \cdots M_{i_k},$
where $i_j \in [1,r]$ for all $1 \le j \le k$, such that $\psi(M) = (0,0,c)$ for some $c \in \mathbb{Q}$, 
then, 
\begin{align*}
c = \sum_{j=1}^k (c_{i_j} - \frac{q}{2}a_{i_j}^2 )
\end{align*}
for some $q \in \mathbb{Q}$, dependent only on $G$.
\end{lemma}

\begin{proof}
Consider the sequence $M_{i_1} M_{i_2} \cdots M_{i_k}$ and let $M_i = \begin{psmallmatrix}
1 & a_i & c_i\\
0 & 1 & b_i\\
0 & 0 & 1
\end{psmallmatrix}$ for each $i \in [1,r]$. Since the superdiagonal vectors 
are parallel, i.e., $a_ib_j=b_ia_j$ for any $i,j\in[1,r]$, we have $q = \frac{b_i}{a_i} \in \mathbb{Q}$ and thus $a_iq = b_i$ for all $i \in [1,r]$. Let us consider the product of the matrices. Then the value~$c$ is equal to
\begin{align*}
c  &=  \sum_{j=1}^k c_{i_j}   +  \sum_{\ell=1}^{k-1} \Bigg( \sum_{j=1}^{\ell}a_{i_j} \Bigg) a_{i_{\ell+1}} q 
=  \sum_{j=1}^k c_{i_j}   + \frac{1}{2} \Bigg( \sum_{\ell=1}^{k}\sum_{j=1}^{k} a_{i_\ell} a_{i_j} q - \sum_{j=1}^k a_{i_j}^2 q \Bigg) \\
 &=  \sum_{j=1}^k (c_{i_j} - \frac{q}{2}a_{i_j}^2).
\end{align*}
The first equality follows from a direct computation as
{\small\begin{align}\label{eq:sum}
\!\!\!\sum_{\ell=1}^k\sum_{j=1}^k a_{i_\ell}a_{i_j} &=\sum_{\ell=1}^{k-1} \Bigg( \sum_{j=1}^{\ell}a_{i_j} \Bigg) a_{i_{\ell+1}}+a_{i_1}(a_{i_1}+\ldots+a_{i_k})+a_{i_2}(a_{i_2}+\ldots+a_{i_k})+\ldots+a_{i_k}a_{i_k} \nonumber \\ 
&=\sum_{\ell=1}^{k-1} \Bigg( \sum_{j=1}^{\ell}a_{i_j} \Bigg) a_{i_{\ell+1}}+a_{i_1}a_{i_1}+a_{i_2}(a_{i_1}+a_{i_2})+\ldots+a_{i_k}(a_{i_1}+\ldots+a_{i_k}) \\
&=\sum_{\ell=1}^{k-1} \Bigg( \sum_{j=1}^{\ell}a_{i_j} \Bigg) a_{i_{\ell+1}}+\sum_{\ell=1}^{k-1} \Bigg( \sum_{j=1}^{\ell}a_{i_j} \Bigg) a_{i_{\ell+1}}+\sum_{j=1}^k a_{i_j}^2. \nonumber
\end{align}}
Note that $\sum_{j=1}^k a_{i_j}=0$ by our choice of the sequence of matrices.
The value $c$ is preserved in case of reordering of matrices due to their commutativity.
\end{proof}

Note that the previous lemma also holds for \heisr.

It is worth mentioning that the identity problem in the Heisenberg group 
is 
decidable if any two matrices have pairwise parallel superdiagonal vectors since 
now the problem reduces to solving a system of two linear homogeneous Diophantine equations.
Hence, it remains to consider the case when there exist two matrices with 
non-parallel superdiagonal vectors in the sequence generating the identity 
matrix. In the following, we prove that the identity matrix is always 
constructible if we can construct any matrix with the zero superdiagonal vector by using matrices with non-parallel superdiagonal vectors.

\begin{lemma}\label{lem:nonparallel}
Let $S = \langle M_1, \ldots, M_r \rangle \subseteq {\rm H}(3,\mathbb{Q})$ be a finitely generated matrix semigroup. Then the identity matrix exists in $S$ if there exists a sequence of matrices
$M_{i_1} M_{i_2} \cdots M_{i_k},$
where $i_j \in [1,r]$ for all $1 \le j \le k$, satisfying the following properties:
\begin{enumerate}[(i)]
\item $\psi(M_{i_1} M_{i_2} \cdots M_{i_k}) = (0,0,c)$ for some $c \in \mathbb{Q}$, and
\item $\vec{v}(M_{i_{j_1}})$ and $\vec{v}(M_{i_{j_2}})$ are not parallel for some $j_1, j_2 \in [1,k]$.
\end{enumerate}
\end{lemma}
\begin{proof}
Let $M = M_{i_1} M_{i_2} \cdots M_{i_k}$ and $\psi(M) = (0,0,c)$ for 
some $c \in \mathbb{Q}$.
If $c=0$, then $M$ is the identity matrix, 
hence we assume that $c > 0$ as the case of $c < 0$ is symmetric.

Given that $M_i$ is the $i$th generator and $\psi(M_i) = (a_i,b_i,c_i)$, 
we have 
$\sum_{j=1}^k a_{i_j} = 0$ and $ \sum_{j=1}^k b_{i_j} = 0$.
Since $c > 0$, the following also holds: 
\begin{align}\label{eq:cvalue}
c = \sum_{\ell=1}^{k-1}\sum_{j=1}^\ell a_{i_j} b_{i_{\ell +1}}  +  \sum_{j=1}^k c_{i_j} >0.
\end{align}

If the matrix semigroup~$S\subseteq \heisq$ has two different matrices $N_1$ and $N_2$ such that $\psi(N_1) = (0,0,c_1)$ and $\psi(N_2) = (0,0,c_2)$ and 
$c_1 c_2 < 0$, then the identity matrix exists in $S$. 
Let $\psi(N_1) = (0,0,\frac{p_1}{q_1})$ and $\psi(N_2) = (0,0,\frac{p_2}{q_2})$, where $p_1,q_1, q_2 \in \mathbb{Z}$ are positive and 
$p_2 \in \mathbb{Z}$ is negative. Then it is easy to see that 
the matrix $N_1^{-q_1p_2}N_2^{q_2p_1}$ exists in $S$ and that 
$\psi(N_1^{-q_1p_2}N_2^{q_2p_1}) = (0,0,0).$

Now we will prove that if $S$ contains a matrix $M$ such that 
$\psi(M) = (0,0,c)$, where $c > 0$, then there also exists a matrix 
$M'$ such that $\psi(M') = (0,0,c')$, where $c' < 0$.

First, we classify the matrices into four types as follows.
A matrix with a superdiagonal vector~$(a,b)$ is classified as
\begin{enumerate}[1)]
\item the $({\scriptstyle +,+})$-type if $a,b >0$,
\item the $({\scriptstyle +,-})$-type if $a\ge 0$ and $b\le0$,
\item the $({\scriptstyle -,-})$-type if $a,b < 0$, and
\item the $({\scriptstyle -,+})$-type if $a<0$ and $b>0$.
\end{enumerate}
Let $G = \{M_1, M_2, \ldots, M_r\}$ be the generating set of the 
matrix semigroup~$S$. Then 
$G = G_{({\scriptscriptstyle +,+})} \sqcup  G_{({\scriptscriptstyle +,-})} \sqcup G_{({\scriptscriptstyle -,-})} \sqcup G_{({\scriptscriptstyle -,+})}$ such that $
G_{(\xi_1,\xi_2)}$ is the set of matrices of the $(\xi_1,\xi_2)$-type, where $\xi_1, \xi_2 \in \{+,-\}$.

Recall that we assume $M = M_{i_1}  \cdots M_{i_k}$ and $\psi(M) = (0,0,c)$ for some $c>0$.
The main idea of the proof is to generate a matrix $M'$ such that $\psi(M') = (0,0,c')$ for some $c' < 0$ by duplicating the matrices in the sequence $M = M_{i_1} \cdots M_{i_k}$ multiple times and reshuffling. Note that any permutation of the sequence 
generating the matrix $M$ such that $\psi(M) = (0,0,c)$ still generates 
matrices $M'$ such that $\psi(M') = (0,0,c')$ since the multiplication of matrices exchanges the first two coordinates in a commutative way. Moreover, we can still 
obtain matrices $M''$ such that $\psi(M'') = (0,0,c'')$ for some 
$c'' \in \mathbb{Q}$ if we shuffle two different permutations 
of the sequence $M_{i_1}\cdots M_{i_k}$ by the same reason.

\begin{figure}[htb]
\centering\begin{tikzpicture}[xscale=0.7,yscale=0.5,every node/.style={scale=1}]
\draw (0,10.5) -- (0,0);
\draw (0,0) -- (17,0);
\draw[pattern=dots, pattern color=blue]  (2,0) node (v13) {} rectangle (3,3);
\draw[pattern=dots, pattern color=blue]  (3,0) node (v15) {} rectangle (4,5.5);
\draw [dashed] (4,0) -- (4,10.5);
\draw [pattern=north east lines, pattern color=red] (4,0) rectangle (5,7.5);
\draw [pattern=north east lines, pattern color=red]  (5,0)  rectangle (6.5,9);
\draw [pattern=north east lines, pattern color=red]  (6.5,0)  rectangle (8,10);
\draw [dashed] (12,10.5) -- (12,0);
\draw (16,0) -- (16,10.5);
\draw [dashed] (8,10.5) -- (8,0);
\draw[pattern=north east lines, pattern color=red]  (8,9) rectangle (10.5,0);
\draw[pattern=north east lines, pattern color=red]  (10.5,5) rectangle (12,0);
\draw [pattern=dots, pattern color=blue] (12,4) rectangle (13,0);
\draw [pattern=dots, pattern color=blue] (13,3) rectangle (14,0);
\draw [pattern=dots, pattern color=blue] (14,2) rectangle (15,0);
\draw [pattern=dots, pattern color=blue] (15,1) rectangle (16,0);
\draw (1,0) -- (1,-1);
\draw (2,0) -- (2,-1);
\draw (3,0) -- (3,-1);
\draw (4,0) -- (4,-1);
\draw (5,0) -- (5,-1);
\draw (6.5,0) -- (6.5,-1);
\draw (8,0) -- (8,-1);
\draw (10.5,0) -- (10.5,-1);
\draw (12,0) -- (12,-1);
\draw (13,0) -- (13,-1);
\draw (14,0) -- (14,-1);
\draw (15,0) -- (15,-1);
\draw (16,0) -- (16,-1);
\draw [<->] (1,-0.5) -- (2,-0.5);
\draw [<->] (2,-0.5) -- (3,-0.5);
\draw [<->] (3,-0.5) -- (4,-0.5);
\draw [<->] (4,-0.5) -- (5,-0.5);
\draw [<->] (5,-0.5) -- (6.5,-0.5);
\draw [<->] (6.5,-0.5) -- (8,-0.5);
\draw [<->] (8,-0.5) -- (10.5,-0.5);
\draw [<->] (10.5,-0.5) -- (12,-0.5);
\draw [<->] (12,-0.5) -- (13,-0.5);
\draw [<->] (13,-0.5) -- (14,-0.5);
\draw [<->] (14,-0.5) -- (15,-0.5);
\draw [<->] (15,-0.5) -- (16,-0.5);
\node at (0.5,-1) {$b_{1}$};
\node at (1.5,-1) {$b_{2}$};
\node at (2.5,-1) {$b_{3}$};
\node at (3.5,-1) {$b_{4}$};
\node at (4.5,-1) {$|b_{5}|$};
\node at (5.75,-1) {$|b_{6}|$};
\node at (7.25,-1) {$|b_{7}|$};
\node at (9.25,-1) {$|b_{8}|$};
\node at (11.25,-1) {$|b_{9}|$};
\node at (12.5,-1) {$b_{{10}}$};
\node at (13.5,-1) {$b_{{11}}$};
\node at (14.5,-1) {$b_{{12}}$};
\node at (15.5,-1) {$b_{{13}}$};
\draw  [pattern=dots, pattern color=blue]  (2,2) rectangle (1,0);
\draw (-1,0) -- (1,0);
\draw (-1,2) -- (0,2);
\draw (-1,3) -- (0,3);
\draw (-1,5.5) -- (0,5.5);
\draw (-1,7.5) -- (0,7.5);
\draw (-1,9) -- (0,9);
\draw (-1,10) -- (0,10);
\draw (16,1) -- (17,1);
\draw (16,2) -- (17,2);
\draw (16,3) -- (17,3);
\draw (16,4) -- (17,4);
\draw (16,5) -- (17,5);
\draw (16,9) -- (17,9);
\draw (16,10) -- (17,10);
\draw [<->]  (-0.5,10) -- (-0.5,9);
\draw [<->]  (-0.5,9) -- (-0.5,7.5);
\draw [<->]  (-0.5,7.5) -- (-0.5,5.5);
\draw [<->] (-0.5,5.5) -- (-0.5,3);
\draw [<->] (-0.5,3) -- (-0.5,2);
\draw[<->]  (-0.5,2) -- (-0.5,0);
\node at (-1,9.5) {$a_6$};
\node at (-1,8.25) {$a_5$};
\node at (-1,6.5) {$a_4$};
\node at (-1,4.25) {$a_3$};
\node at (-1,2.5) { $a_2$};
\node at (-1,1) {$a_1$};
\draw [<->](16.5,10) -- (16.5,9);
\draw[<->] (16.5,9) -- (16.5,5);
\draw [<->](16.5,5) -- (16.5,4);
\draw [<->] (16.5,4) -- (16.5,3);
\draw [<->] (16.5,3) -- (16.5,2);
\draw [<->] (16.5,2) -- (16.5,1);
\draw [<->] (16.5,1) -- (16.5,0);
\node at (17.3,9.5) {$|a_7|$};
\node at (17.3,7) {$|a_8|$};
\node at (17.3,4.5) {$|a_9|$};
\node at (17.3,3.5) {$|a_{10}|$};
\node at (17.3,2.5) {$|a_{11}|$};
\node at (17.3,1.5) {$|a_{12}|$};
\node at (17.3,0.5) {$|a_{13}|$};
\draw [draw=none,pattern=dots, pattern color=blue] (0.4,10.5) rectangle (1.15,9.75);
\draw [draw=none,pattern=north east lines, pattern color=red] (0.4,9.5) rectangle (1.15,8.75);
\node [align=left] at (2.4,10.1) {: positive};
\node [align=left] at (2.49,9.1) {: negative};
\draw (0,0) -- (0,-1);
\draw [<->] (0,-0.5) -- (1,-0.5);
\end{tikzpicture}
\caption{The histogram describes how the upper-right corner of $M_{1}\cdots M_{13}$ is computed by multiplications. The blue dotted (red lined) area implies the value which will be added to (subtracted from) the upper-right corner of the final matrix after multiplications of matrices in the sequence.}
\label{fig:example}
\end{figure}
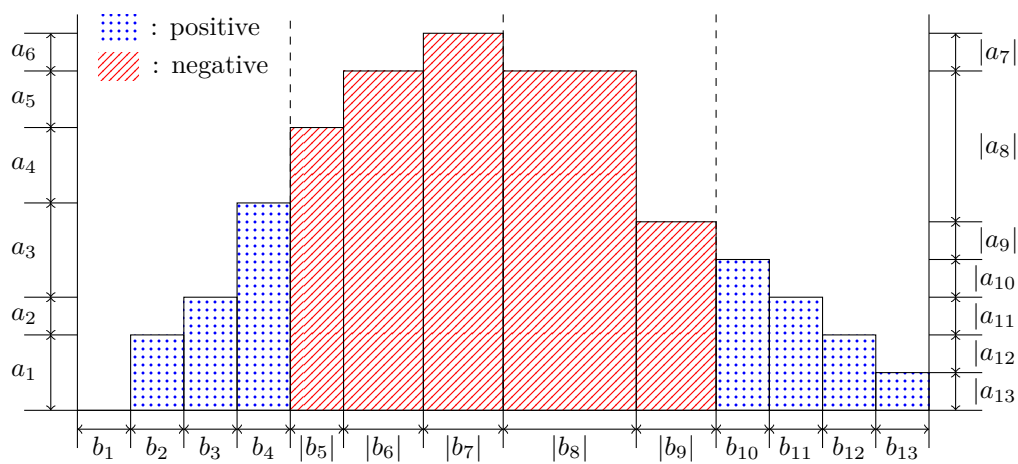
\begin{figure}[htb]
\centering
\begin{tikzpicture}[xscale=0.6,yscale=0.5,every node/.style={scale=1.0}]
\draw (0,11) -- (0,0);
\draw (0,0) -- (16,0);
\draw [dashed] (4,0) -- (4,11);
\draw [dashed] (12,0) -- (12,11);
\draw (16,0)  -- (16,11);
\draw [dashed] (8,0) -- (8,11);
\draw [draw=none,pattern=dots, pattern color=blue] (0.25,10.5) rectangle (1,9.75);
\draw [draw=none,pattern=north east lines, pattern color=red] (0.25,9.5) rectangle (1,8.75);
\node [align=left] at (2.25,10.1) {: positive};
\node [align=left] at (2.35,9.1) {: negative};
\draw  [pattern=dots, pattern color=blue] (0,1) rectangle (0.5,0);
\draw  [pattern=dots, pattern color=blue]  (0.5,0) rectangle (1,2);
\draw   [pattern=dots, pattern color=blue] (1,0) rectangle (1.5,3);
\draw (0,0) -- (4,8) -- (8,10);
\draw (8,10) -- (12,4) -- (12,4) -- (16,0);
\draw  [pattern=dots, pattern color=blue]  (1.5,4) rectangle (2,0);
\draw  [pattern=dots, pattern color=blue]  (2,5) rectangle (2.5,0);
\draw  [pattern=dots, pattern color=blue]  (2.5,6) rectangle (3,0);
\draw  [pattern=dots, pattern color=blue]  (3,7) rectangle (3.5,0);
\draw  [pattern=dots, pattern color=blue]  (3.5,8) rectangle (4,0);
\draw [pattern=north east lines, pattern color=red] (4,8.0) rectangle (4.5,0);
\draw [pattern=north east lines, pattern color=red]  (4.5,8.25) rectangle (5,0);
\draw [pattern=north east lines, pattern color=red]  (5,8.5) rectangle (5.5,0);
\draw [pattern=north east lines, pattern color=red]  (5.5,8.75) rectangle (6,0);
\draw [pattern=north east lines, pattern color=red]  (6,9) rectangle (6.5,0);
\draw [pattern=north east lines, pattern color=red]  (6.5,9.25) rectangle (7,0);
\draw [pattern=north east lines, pattern color=red]  (7,9.5) rectangle (7.5,0);
\draw [pattern=north east lines, pattern color=red]  (7.5,9.75) rectangle (8,0) ;
\draw [pattern=north east lines, pattern color=red]  (8,9.25) rectangle (8.5,0);
\draw [pattern=north east lines, pattern color=red]  (8.5,8.5) rectangle (9,0);
\draw [pattern=north east lines, pattern color=red]  (9,7.75) rectangle (9.5,0);
\draw [pattern=north east lines, pattern color=red]  (9.5,7) rectangle (10,0);
\draw [pattern=north east lines, pattern color=red]  (10,6.25) rectangle (10.5,0);
\draw [pattern=north east lines, pattern color=red]  (10.5,5.5) rectangle (11,0);
\draw [pattern=north east lines, pattern color=red]  (11,4.75) rectangle (11.5,0);
\draw [pattern=north east lines, pattern color=red]  (11.5,4) rectangle (12,0);
\draw [pattern=dots, pattern color=blue]   (12,4) rectangle (12.5,0);
\draw [pattern=dots, pattern color=blue]   (12.5,3.5) rectangle (13,0);
\draw [pattern=dots, pattern color=blue]   (13,3) rectangle (13.5,0);
\draw [pattern=dots, pattern color=blue]   (13.5,2.5) rectangle (14,0);
\draw [pattern=dots, pattern color=blue]   (14,2) rectangle (14.5,0);
\draw  [pattern=dots, pattern color=blue]  (14.5,1.5) rectangle (15,0);
\draw [pattern=dots, pattern color=blue]   (15,1) rectangle (15.5,0);
\draw [pattern=dots, pattern color=blue]   (15.5,0.5) rectangle (16,0);
\draw  (0,0) -- (0,-1);
\draw  (4,0) -- (4,-1);
\draw (8,0) -- (8,-1);
\draw (12,0) -- (12,-1);
\draw (16,0) -- (16,-1);
\draw [<->] (0,-0.5) -- (4,-0.5);
\draw [<->] (8,-0.5) -- (4,-0.5);
\draw [<->](8,-0.5) -- (12,-0.5);
\draw [<->] (12,-0.5) -- (16,-0.5);
\draw (-1,0) -- (0,0);
\draw (-1,8) -- (0,8);
\draw (-1,10) -- (0,10);
\draw (16,0) -- (17,0);
\draw (16,4) -- (17,4);
\draw [<->](-0.5,10) -- (-0.5,8);
\draw [<->](-0.5,8) -- (-0.5,0);
\draw [<->](16.5,4) -- (16.5,0);
\draw (16,10) -- (17,10);
\draw [<->](16.5,10) -- (16.5,4);
\node at (2,-1) {$b_{({\scriptscriptstyle +,+})}m$};
\node at (6,-1) {$|b_{({\scriptscriptstyle +,-})}|m$};
\node at (10,-1) {$|b_{({\scriptscriptstyle -,-})}|m$};
\node at (14,-1) {$b_{({\scriptscriptstyle -,+})}m$};
\node at (18,2) {$|a_{({\scriptscriptstyle -,+})}|m$};
\node at (18,7) {$|a_{({\scriptscriptstyle -,-})}|m$};
\node at (-1.7,4) {$a_{({\scriptscriptstyle +,+})}m$};
\node at (-1.7,9) {$a_{({\scriptscriptstyle +,-})}m$};
\end{tikzpicture}
\caption{The histogram describes how the value in the upper-right corner of matrix $M_{({\scriptscriptstyle +,+})}^m M_{({\scriptscriptstyle +,-})}^m M_{({\scriptscriptstyle -,-})}^m M_{({\scriptscriptstyle -,+})}^m$ is computed by multiplications. Here $m = 8$.}
\label{fig:repeat}
\end{figure}
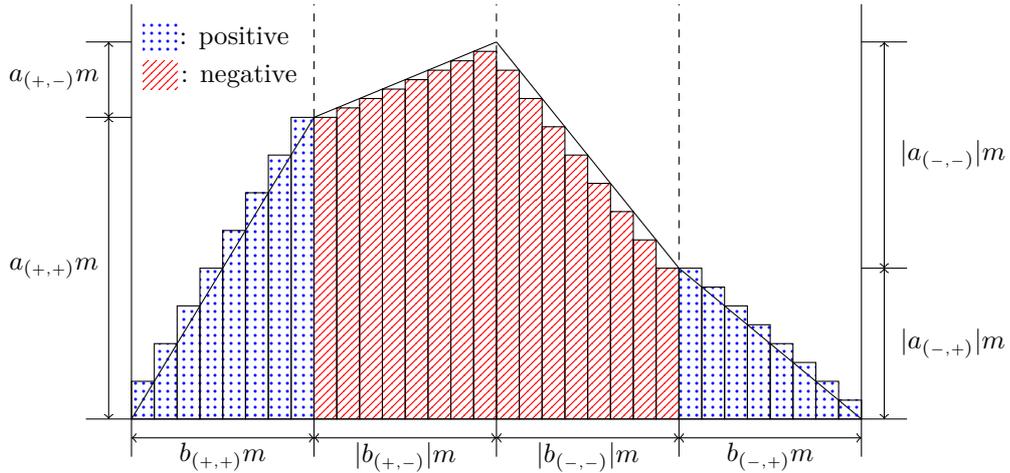

Let us illustrate the idea with the following example. See 
Figure~\ref{fig:example} and Figure~\ref{fig:repeat} for pictorial descriptions 
of the idea.
Let $\{M_i \mid 1\le i\le 4\} \subseteq G_{({\scriptscriptstyle +,+})}$, $\{M_i \mid 5\le i\le 7\} \subseteq G_{({\scriptscriptstyle +,-})}$, $\{M_i \mid 8\le i\le 9\} \subseteq G_{({\scriptscriptstyle -,-})}$, and $\{M_i \mid 10\le i\le 13\} \subseteq G_{({\scriptscriptstyle -,+})}$. Then assume that $M_1M_2\cdots M_{13}=\begin{psmallmatrix}1&0&x\\0&1&0\\0&0&1\end{psmallmatrix}$,
where $x$ is computed by \eqref{eq:cvalue}. As we mentioned above, 
$x$ changes if we change the order of multiplicands. 
In this example, we first multiply 
$({\scriptstyle +,+})$-type matrices and accumulate the values in the superdiagonal 
coordinates since these matrices have positive values in the coordinates. 
Indeed, the blue dotted area implies the value we add to the upper-right 
corner by multiplying such matrices. Then we multiply $({\scriptstyle +,-})$-type 
matrices and still increase the `$a$'-value. The `$b$'-values in 
$({\scriptstyle +,-})$-type matrices are negative thus, the red lined area is 
subtracted from the upper-right corner. We still subtract by multiplying 
$({\scriptstyle -,-})$-type matrices since the accumulated `$a$'-value is still positive 
and `$b$'-values are negative. Then we finish the multiplication by adding 
exactly the last blue dotted area to the upper-right corner. It is easy to see that the total subtracted value is larger than the total added value.

However, we cannot guarantee that $x$ is negative since $\sum_{i=1}^{13} c_i$ could be larger than the contribution from the superdiagonal coordinates.
This is why we need to copy the sequence of matrices generating the matrix 
corresponding to the triple~$(0,0,c)$ for some $c \in \mathbb{Q}$. 
In Figure~\ref{fig:repeat}, we describe an example where we duplicate the sequence eight times and shuffle and permute them in order to minimize the 
value in the upper-right corner. Now the lengths of both axes are $m$ ($m=8$ in this example) times larger than before and it follows that the area also grows quadratically in $m$. Since the summation 
$m \cdot \sum_{i=1}^{13} c_i$ grows linearly in $m$, we have $x <0$ when $m$ is large enough.

For each $\xi_1,\xi_2\in\{+,-\}$, let us define multisets $S_{(\xi_1, \xi_2)}$  that are obtained from the sequence~$M_{i_1} \cdots M_{i_k}$ by partitioning the product according to the matrix types. 
That is, $S_{(\xi_1, \xi_2)}$ contains exactly the matrices of $(\xi_1,\xi_2)$-type in the product (possibly with several copies of each matrix).

For each $\xi_1, \xi_2 \in \{+,-\}$, let us define $a_{(\xi_1,\xi_2)},b_{(\xi_1,\xi_2)},c_{(\xi_1,\xi_2)}$ such that
\begin{align*}
(a_{(\xi_1, \xi_2)} ,b_{(\xi_1, \xi_2)}, c_{(\xi_1, \xi_2)})=\sum_{M \in S_{(\xi_1, \xi_2)}} \psi(M).
\end{align*}
In other words, $a_{(\xi_1, \xi_2)}$ ($b_{(\xi_1, \xi_2)}$ and $c_{(\xi_1, \xi_2)}$, respectively) is 
the sum of the values in the `$a$' (`$b$' and `$c$', respectively) coordinate from the matrices in the multiset $S_{(\xi_1, \xi_2)}$.

Now consider a permutation of the sequence $M_{i_1}\cdots M_{i_k}$, where the first part of the sequence only consists of the $({\scriptstyle +,+})$-type matrices, the second 
part only consists of the $({\scriptstyle +,-})$-type matrices, the third part 
only consists of the $({\scriptstyle -,-})$-type, and finally the last part 
only consists of the $({\scriptstyle -,+})$-type. 

Let us denote by 
$M_{({\scriptscriptstyle +,+})}$ the matrix which results from the multiplication of the first part, namely, $M_{({\scriptscriptstyle +,+})} = \prod_{M \in S_{({\scriptscriptstyle +,+})}} M.$
Then $\psi(M_{({\scriptscriptstyle +,+})}) =  (a_{({\scriptscriptstyle +,+})},  b_{(+, +)}, x_{({\scriptscriptstyle +,+})})$ holds,
where $x_{({\scriptscriptstyle +,+})} < c_{(+, +)} + a_{({\scriptscriptstyle +,+})} b_{({\scriptscriptstyle +,+})}.$
Let us define $M_{({\scriptscriptstyle +,-})}$, $M_{({\scriptscriptstyle -,-})}$ and $M_{({\scriptscriptstyle -,+})}$ in a similar fashion. Note that for $M_{({\scriptscriptstyle +,-})}$ and $M_{({\scriptscriptstyle -,+})}$, the term $x$ is bounded from below.

Now we claim that there exists an integer~$m > 0$ such that $M_{({\scriptscriptstyle +,+})}^m M_{({\scriptscriptstyle +,-})}^m M_{({\scriptscriptstyle -,-})}^m M_{({\scriptscriptstyle -,+})}^m$ corresponds to the triple~$(0,0,c')$ for some $c' < 0$. Let $N$ be a matrix in $ {\rm H}(3,\mathbb{Q})$
and $\psi(N) = (a,b,c)$. Then the upper-triangular coordinates of the 
$m$th power of $N$ are calculated as follows:
$\psi(N^m) = (am, bm, cm  + ab \cdot \frac{1}{2}m(m-1))$.

Next, we consider how the upper-triangular coordinates are affected by multiplication of matrices $M_{({\scriptscriptstyle +,+})}^m$, $M_{({\scriptscriptstyle +,-})}^m$, $M_{({\scriptscriptstyle -,-})}^m$ and $M_{({\scriptscriptstyle -,+})}^m$. Let us consider the first part of the product, $M_{({\scriptscriptstyle +,+})}^m$, that is,
$\psi(M_{({\scriptscriptstyle +,+})}^m) = (a_{({\scriptscriptstyle +,+})}m, b_{({\scriptscriptstyle +,+})}m, x_{({\scriptscriptstyle +,+})}m + z_1)$,
where 
$z_1$ can be found in Table~\ref{tab:z1234}.
Now we multiply $M_{({\scriptscriptstyle +,+})}^m$ by the second part $M_{({\scriptscriptstyle +,-})}^m$. Then the resulting matrix $M_{({\scriptscriptstyle +,+})}^mM_{({\scriptscriptstyle +,-})}^m$ corresponds to 
\begin{align*}
\psi(M_{({\scriptscriptstyle +,+})}^mM_{({\scriptscriptstyle +,-})}^m)=((a_{({\scriptscriptstyle +,+})}+a_{({\scriptscriptstyle +,-})})m, (b_{({\scriptscriptstyle +,+})}+b_{({\scriptscriptstyle +,-})})m, (x_{({\scriptscriptstyle +,+})}+x_{({\scriptscriptstyle +,-})})m + z_1 - z_2), 
\end{align*}
where
$z_2$ can be found in Table~\ref{tab:z1234}.
Similarly, we compute $z_3$ and $z_4$ that will be added to the upper-right corner as a result of multiplying $M_{({\scriptscriptstyle -,-})}^m$ and $M_{({\scriptscriptstyle -,+})}^m$ and present them in Table~\ref{tab:z1234}.

\begin{table}[htb]
\begin{align*}
z_1 &= | a_{({\scriptscriptstyle +,+})}|| b_{({\scriptscriptstyle +,+})}| \cdot \frac{1}{2}m(m-1), \\
z_2 &= m^2 |a_{({\scriptscriptstyle +,+})}| |b_{({\scriptscriptstyle +,-})}| + |a_{({\scriptscriptstyle +,-})} ||b_{({\scriptscriptstyle +,-})}| \cdot \frac{1}{2}m(m-1), \\
z_3 &=  |a_{({\scriptscriptstyle -,+})}|| b_{({\scriptscriptstyle -,-})}|  m^2 + |a_{({\scriptscriptstyle -,-})} || b_{({\scriptscriptstyle -,-})}| \cdot \frac{1}{2}m(m-1) \ \text{and} \\
z_4 &= |a_{({\scriptscriptstyle -,+})}| |b_{({\scriptscriptstyle -,+})}| \cdot  \frac{1}{2}m(m-1).
\end{align*}
\caption{Values $z_1$, $z_2$, $z_3$ and $z_4$ in the product $M_{({\scriptscriptstyle +,+})}^m M_{({\scriptscriptstyle +,-})}^m M_{({\scriptscriptstyle -,-})}^m M_{({\scriptscriptstyle -,+})}^m$ \label{tab:z1234}}
\end{table}

After the multiplying all four parts, we have
\begin{multline*}
\psi(M_{({\scriptscriptstyle +,+})}^m M_{({\scriptscriptstyle +,-})}^m M_{({\scriptscriptstyle -,-})}^m M_{({\scriptscriptstyle -,+})}^m) = \\ 
(0,0, (x_{({\scriptscriptstyle +,+})}+x_{({\scriptscriptstyle +,-})} + x_{({\scriptscriptstyle -,-})} + x_{({\scriptscriptstyle -,+})})m + z_1 - z_2 - z_3 + z_4).
\end{multline*}

Denote $z = z_1 - z_2 - z_3 + z_4$. From the above equations, we can see that $z$ can be represented as a quadratic equation of $m$ and that the coefficient of $m^2$ is always negative if $S_{(\xi_1, \xi_2)} \ne \emptyset$ for all $\xi_1, \xi_2 \in \{ +,-\}$. 
That is, the coefficient of $m^2$ is
\begin{multline*}
\frac{1}{2}(|a_{({\scriptscriptstyle +,+})}|| b_{({\scriptscriptstyle +,+})}| + |a_{({\scriptscriptstyle -,+})}|| b_{({\scriptscriptstyle -,+})}|) - \frac{1}{2}(|a_{({\scriptscriptstyle +,-})}|| b_{({\scriptscriptstyle +,-})}| + |a_{({\scriptscriptstyle -,-})}|| b_{({\scriptscriptstyle -,-})}|) \\ 
{} + |a_{({\scriptscriptstyle +,+})}|| b_{({\scriptscriptstyle +,-})}| + |a_{({\scriptscriptstyle -,+})}|| b_{({\scriptscriptstyle -,-})}|. 
\end{multline*}

Let us simplify the equation by denoting $|a_{({\scriptscriptstyle +,+})}| + |a_{({\scriptscriptstyle +,-})}| = |a_{({\scriptscriptstyle -,+})}| + |a_{({\scriptscriptstyle -,-})}| = a'$ and $|b_{({\scriptscriptstyle +,+})}| + |b_{({\scriptscriptstyle -,+})}| = |b_{({\scriptscriptstyle +,-})}| + |b_{({\scriptscriptstyle -,-})}| = b'$. Note, that the equations hold as we are considering the product, where `a' and `b' elements add up to zero.  Then 
\begin{align*}
a'b' &= a' (|b_{({\scriptscriptstyle +,-})}| + |b_{({\scriptscriptstyle -,-})}|) = a' |b_{({\scriptscriptstyle +,-})}| + a'|b_{({\scriptscriptstyle -,-})}| \\
&= (|a_{({\scriptscriptstyle +,+})}| + |a_{({\scriptscriptstyle +,-})}|)|b_{({\scriptscriptstyle +,-})}| + (|a_{({\scriptscriptstyle -,+})}| + |a_{({\scriptscriptstyle -,-})}|)|b_{({\scriptscriptstyle -,-})}|.
\end{align*}
Now the coefficient of $m^2$ in $z$ can be written as
\begin{align}\label{eq:negative}
-a'b' + \frac{1}{2}(|a_{({\scriptscriptstyle +,+})}|| b_{({\scriptscriptstyle +,+})}| + |a_{({\scriptscriptstyle -,+})}|| b_{({\scriptscriptstyle -,+})}| +  |a_{({\scriptscriptstyle -,-})}|| b_{({\scriptscriptstyle -,-})}| + |a_{({\scriptscriptstyle +,-})}|| b_{({\scriptscriptstyle +,-})}|  ).
\end{align}

Without loss of generality, suppose that $|a_{({\scriptscriptstyle +,+})}| \ge |a_{({\scriptscriptstyle -,+})}|$. Then we have
\begin{align*}
\!\!\!\!\!|a_{({\scriptscriptstyle +,+})}|| b_{({\scriptscriptstyle +,+})}| + |a_{({\scriptscriptstyle -,+})}|| b_{({\scriptscriptstyle -,+})}| \le |a_{({\scriptscriptstyle +,+})}|b'
\mbox{ and }
|a_{({\scriptscriptstyle -,-})}|| b_{({\scriptscriptstyle -,-})}| + |a_{({\scriptscriptstyle +,-})}|| b_{({\scriptscriptstyle +,-})}| \le |a_{({\scriptscriptstyle -,-})}|b'.
\end{align*}
From
$
(|a_{({\scriptscriptstyle +,+})}|+ |a_{({\scriptscriptstyle -,-})}|)b' \le 2a'b',
$
we can see that the coefficient of the highest 
power of the variable is negative in $z$ if $|a_{({\scriptscriptstyle +,+})}|+ |a_{({\scriptscriptstyle -,-})}| < 2a'$.
By comparing two terms in \eqref{eq:negative}, we can see that the coefficient is negative if all subsets $S_{({\scriptscriptstyle -,+})}$, $S_{({\scriptscriptstyle +,-})}$, $S_{({\scriptscriptstyle +,+})}$ and $S_{({\scriptscriptstyle -,-})}$ are not empty.
Since the coefficient of the highest power of the variable is negative, $z$ becomes negative when $m$ is large enough. Therefore, we have a matrix 
corresponding to the triple~$(0,0,c')$ for some $c'<0$ as a product of 
multiplying matrices in the generating set and the identity matrix is 
also reachable.


It should be noted that there are some subcases where some of 
subsets from $S_{({\scriptscriptstyle +,+})}$, $S_{({\scriptscriptstyle -,+})}$, $S_{({\scriptscriptstyle +,-})}$, and 
$S_{({\scriptscriptstyle -,-})}$ are empty. We examine all possible cases and prove that 
the coefficient of $m^2$ is negative in every case and the matrix 
with a negative number in the corner is constructible.
First, we prove that the coefficient of $m^2$ in $z$ is negative 
when only one of the subsets from  
$S_{({\scriptscriptstyle +,+})}$.  $S_{({\scriptscriptstyle +,-})}$,  $S_{({\scriptscriptstyle -,-})}$, and  $S_{({\scriptscriptstyle -,+})}$ is empty as 
follows:

Assume that only $S_{({\scriptscriptstyle +,+})} = \emptyset$. In this case, note that $|a_{({\scriptscriptstyle +,-})}| = a'$ and $|b_{({\scriptscriptstyle -,+})}| = b'$ since $|a_{({\scriptscriptstyle +,+})}| = |b_{({\scriptscriptstyle +,+})}| = 0$ by $S_{({\scriptscriptstyle +,+})} = \emptyset$ being empty. Then the coefficient of $m^2$ becomes 
\begin{align*}
-a'b' + \frac{ |a_{({\scriptscriptstyle -,+})}| b' +  |a_{({\scriptscriptstyle -,-})}|| b_{({\scriptscriptstyle -,-})}| + a' | b_{({\scriptscriptstyle +,-})}|}{2}.
\end{align*}
We can see that the coefficient can be at most 
0 since $|a_{({\scriptscriptstyle -,+})}| b' $ and $|a_{({\scriptscriptstyle -,-})}|| b_{({\scriptscriptstyle -,-})}| + a' | b_{({\scriptscriptstyle +,-})}|$ can be maximized to $a'b'$. If we maximize 
$|a_{({\scriptscriptstyle -,+})}| b' $ by setting $|a_{({\scriptscriptstyle -,+})}| = a'$, then $|a_{({\scriptscriptstyle -,-})}|$ is 0 since $|a_{({\scriptscriptstyle +,+})}| + |a_{({\scriptscriptstyle -,+})}| = a'$. 
Then $|a_{({\scriptscriptstyle -,-})}|| b_{({\scriptscriptstyle -,-})}| + a' | b_{({\scriptscriptstyle +,-})}|$ can be $a'b'$ only when $| b_{({\scriptscriptstyle +,-})}| = b'$. This 
leads to the set $S_{({\scriptscriptstyle -,-})}$ being empty since we have $| a_{({\scriptscriptstyle -,-})}| = 0$ and $| b_{({\scriptscriptstyle -,-})}| = 0$ and therefore, 
we have a contradiction.

The remaining cases, $S_{({\scriptscriptstyle +,-})} = \emptyset$, or $S_{({\scriptscriptstyle -,-})} = \emptyset$, or  $S_{({\scriptscriptstyle -,+})} = \emptyset$ are proven analogously.

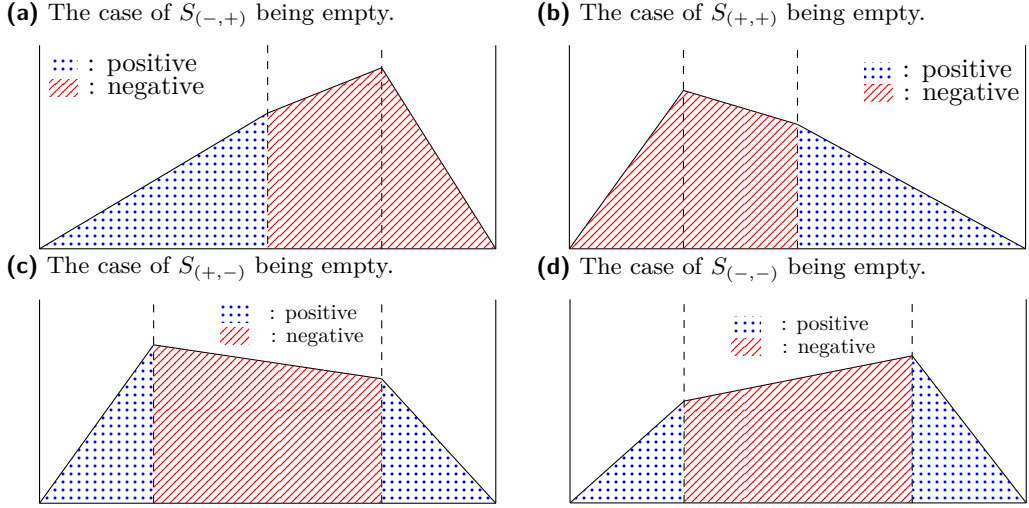
\begin{figure}[htb]
     \begin{subfigure}{0.49\textwidth}     
     \centering
\caption{The case of $S_{(-,+)}$ being empty.}
     \begin{tikzpicture}[xscale=0.5,yscale=0.3,every node/.style={scale=1.0}]
\tikzset{>=latex}
\draw (0,9) -- (0,0);
\draw (0,0) -- (12,0);
\draw [dashed] (6,9) -- (6,0);
\draw (12,0) -- (12,9);
\draw [dashed] (9,0) -- (9,9);
\draw [draw=none,pattern=dots, pattern color=blue] (0.25,8.5) rectangle (1,7.75);
\draw [draw=none,pattern=north east lines, pattern color=red] (0.25,7.5) rectangle (1,6.75);
\node at (2.7,8.1) {: positive};
\node  at (2.8,7.1) {: negative};
\draw (0,0) -- (6,6) -- (9,8);
\draw (9,8) --  (12,0);

\fill[pattern=dots, pattern color=blue] (0,0)--(6,6)--(6,0)--cycle;
\fill[pattern=north east lines, pattern color=red] (6,0)--(6,6)--(9,8)--(9,0)--cycle;
\fill[pattern=north east lines, pattern color=red] (9,0)--(9,8)--(12,0)--cycle;
\end{tikzpicture}
\end{subfigure}
\begin{subfigure}{0.49\textwidth}
\centering
\caption{The case of $S_{(+,+)}$ being empty.}
     \begin{tikzpicture}[xscale=0.5,yscale=0.3,every node/.style={scale=1.0}]
\tikzset{>=latex}
\draw (0,9) -- (0,0);
\draw (0,0) -- (12,0);
\draw [dashed] (3,0) -- (3,9);
\draw (12,0) -- (12,9);
\draw [dashed] (6,0) -- (6,9);
\draw [draw=none,pattern=dots, pattern color=blue] (7.75,8.25) rectangle (8.5,7.5);
\draw [draw=none,pattern=north east lines, pattern color=red] (7.75,7.25) rectangle (8.5,6.5);
\node  at (10.2,7.85) {: positive};
\node  at (10.3,6.85) {: negative};
\draw (0,0) -- (3,7) -- (6,5.5);
\draw (6,5.5) --  (12,0);
%
%
\fill[pattern=north east lines, pattern color=red] (0,0)--(3,7)--(3,0)--cycle;
\fill[pattern=north east lines, pattern color=red] (3,0)--(3,7)--(6,5.5)--(6,0)--cycle;
\fill[pattern=dots, pattern color=blue] (6,0)--(6,5.5)--(12,0)--cycle;
\end{tikzpicture}
\end{subfigure}
\\
\begin{subfigure}{0.49\textwidth}
\centering
\caption{The case of $S_{(+,-)}$ being empty.}
      \begin{tikzpicture}[xscale=0.5,yscale=0.3,every node/.style={scale=0.8}]
\tikzset{>=latex}
\draw (0,9) -- (0,0);
\draw (0,0) -- (12,0);
\draw [dashed] (3,0) -- (3,9);
\draw (12,0) -- (12,9);
\draw [dashed] (9,0) -- (9,9);
\draw [draw=none,pattern=dots, pattern color=blue] (4.75,8.75) rectangle (5.5,8);
\draw [draw=none,pattern=north east lines, pattern color=red] (4.75,7.75) rectangle (5.5,7);
\node at (7.2,8.35) {: positive};
\node at (7.3,7.35) {: negative};
\draw (0,0) -- (3,7) -- (9,5.5);
\draw (9,5.5) --  (12,0);
%
%

\fill[pattern=dots, pattern color=blue] (0,0)--(3,7)--(3,0)--cycle;
\fill[pattern=north east lines, pattern color=red] (3,0)--(3,7)--(9,5.5)--(9,0)--cycle;
\fill[pattern=dots, pattern color=blue] (9,0)--(9,5.5)--(12,0)--cycle;
\end{tikzpicture}
\end{subfigure}
\begin{subfigure}{0.49\textwidth}
\centering
\caption{The case of $S_{(-,-)}$ being empty.}
\begin{tikzpicture}[xscale=0.5,yscale=0.3,every node/.style={scale=0.8}]
\tikzset{>=latex}
\draw (0,9) -- (0,0);
\draw (0,0) -- (12,0);
\draw [dashed] (3,0) -- (3,9);
\draw (12,0)  -- (12,9);
\draw [dashed] (9,0) -- (9,9);
\draw [draw=none,pattern=dots, pattern color=blue] (4.25,8.25) rectangle (5,7.5);
\draw [draw=none,pattern=north east lines, pattern color=red] (4.25,7.25) rectangle (5,6.5);
\node [align=left] at (6.7,7.85) {: positive};
\node [align=left] at (6.8,6.85) {: negative};
\draw (0,0) -- (3,4.5) -- (9,6.5);
\draw (9,6.5) --  (12,0);
%
%
\fill[pattern=dots, pattern color=blue] (0,0)--(3,4.5)--(3,0)--cycle;
\fill[pattern=north east lines, pattern color=red] (3,0)--(3,4.5)--(9,6.5)--(9,0)--cycle;
\fill[pattern=dots, pattern color=blue] (9,0)--(9,6.5)--(12,0)--cycle;
\end{tikzpicture}
\end{subfigure}
\caption{Subcases where one of the subsets from $S_{(+,+)}$, $S_{(-,+)}$, $S_{(+,-)}$, and $S_{(-,-)}$ is empty.}
\label{fig:oneempty}
\end{figure}

Figure~\ref{fig:oneempty} shows the cases when one of subsets from $S_{(+,+)}$, $S_{(-,+)}$, $S_{(+,-)}$, and 
$S_{(-,-)}$ is empty.
Lastly, it remains to consider the cases where two of the subsets are empty. Note that we 
do not consider the cases where three of the subsets are empty because the sum of $a$'s and 
$b$'s cannot be both zero in such cases. 
Here we assume one of $S_{({\scriptscriptstyle +,+})}$ and $S_{({\scriptscriptstyle -,-})}$ contains two matrices whose superdiagonal vectors are not parallel by the statement of this lemma. Then we 
can always make the negative contribution larger by using matrices with 
different superdiagonal vectors. See Figure~\ref{fig:twoempty} for an example.
More formally, we consider the two cases as follows:

Assume first that $S_{({\scriptscriptstyle +,+})} = \emptyset$ and $S_{({\scriptscriptstyle -,-})} = \emptyset$. Without loss of 
generality, assume that $S_{({\scriptscriptstyle -,+})}$ contains two matrices~$M_1$ and $M_2$ with non-parallel superdiagonal vectors. 
Let $\vec{v}(M_1) = (a_1,b_1)$ and $\vec{v}(M_2) = (a_2,b_2)$ be superdiagonal vectors for $M_1$ and $M_2$, respectively, such that $|\frac{a_1}{b_1}| > |\frac{a_2}{b_2}|$. To simplify the proof, we assume 
the set $S_{({\scriptscriptstyle +,-})}$ only uses one matrix~$M_3$, where $\vec{v}(M_3) = (a_3,b_3)$, to generate 
a matrix with a zero superdiagonal vector. This implies that $a_1 x + a_2 y + a_3 = 0$ and 
$b_1 x + b_2 y + b_3 = 0$ for some $x,y \in \mathbb{Q}$. Here the idea is that we first multiply 
the matrix~$M_1$ and then multiply $M_2$ later. For instance, we first multiply $M_1^m$ and 
then $M_2^m$. Then the coefficient of the highest power 
in $z$ becomes $\frac{- a'b' + 2|a_2||b_1| + |a_1||b_1| + |a_2||b_2|}{2}$. Since $a' = |a_1| + |a_2|$ and 
$b' = |b_1| + |b_2|$, the coefficient of $m^2$ is now $\frac{|a_2||b_1| - |a_1||b_2|}{2}$. 
By the supposition $|\frac{a_1}{b_1}| > |\frac{a_2}{b_2}|$, we prove that the coefficient of the 
highest power in $z$ is always negative.

The second case, where $S_{({\scriptscriptstyle +,-})} = \emptyset$ and $S_{({\scriptscriptstyle -,+})} = \emptyset$, is proven analogously.

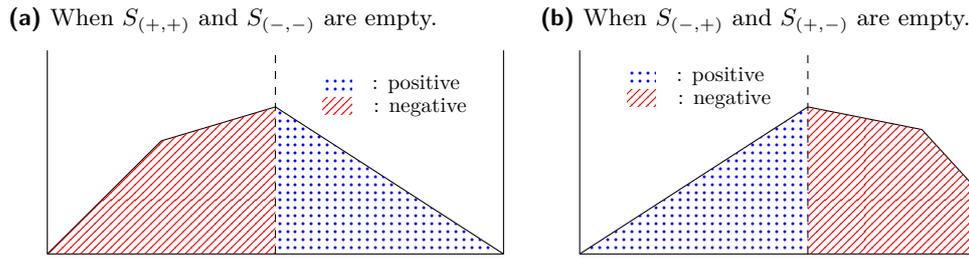
\begin{figure}[htb]
\centering
\begin{subfigure}{0.50\textwidth}
\centering
\caption{When $S_{(+,+)}$ and $S_{(-,-)}$ are empty.}
 \begin{tikzpicture}[xscale=0.5,yscale=0.3,every node/.style={scale=0.8}]
\tikzset{>=latex}
\draw (0,9) -- (0,0);
\draw (0,0) -- (12,0);
\draw (12,0) -- (12,9);
\draw [dashed] (6,0) -- (6,9);
\draw [draw=none,pattern=dots, pattern color=blue] (7.25,7.9) rectangle (8,7.15);
\draw [draw=none,pattern=north east lines, pattern color=red] (7.25,6.9) rectangle (8,6.15);
\node [align=left] at (9.7,7.5) {: positive};
\node [align=left] at (9.8,6.5) {: negative};
\draw (0,0) -- (3,5) -- (6,6.5);
\draw (6,6.5) --  (12,0);
%
%
\fill[pattern=north east lines, pattern color=red] (0,0)--(3,5)--(3,0)--cycle;
\fill[pattern=north east lines, pattern color=red] (3,0)--(3,5)--(6,6.5)--(6,0)--cycle;
\fill[pattern=dots, pattern color=blue] (6,0)--(6,6.5)--(12,0)--cycle;
\end{tikzpicture}
\end{subfigure}\begin{subfigure}{0.50\textwidth}
\centering
\caption{When $S_{(-,+)}$ and $S_{(+,-)}$ are empty.}
  \begin{tikzpicture}[xscale=0.5,yscale=0.3,every node/.style={scale=0.8}]
\tikzset{>=latex}
\draw (0,9) -- (0,0);
\draw (0,0) -- (12,0);
\draw (12,0) -- (12,9);
\draw [dashed] (6,0) -- (6,9);
\draw [draw=none,pattern=dots, pattern color=blue] (1.25,8.25) rectangle (2,7.5);
\draw [draw=none,pattern=north east lines, pattern color=red] (1.25,7.25) rectangle (2,6.5);
\node [align=left] at (3.7,7.85) {: positive};
\node [align=left] at (3.8,6.85) {: negative};
\draw (0,0) --  (6,6.5);
\draw (6,6.5) -- (9,5.5) -- (12,0);
%
%
%

\fill[pattern=dots, pattern color=blue] (0,0)--(6,0)--(6,6.5)--cycle;
\fill[pattern=north east lines, pattern color=red] (6,0)--(6,6.5)--(9,5.5)--(12,0)--cycle;
\end{tikzpicture}
\end{subfigure}
\caption{Subcases where two of the subsets from $S_{(+,+)}$, $S_{(-,+)}$, $S_{(+,-)}$, and 
$S_{(-,-)}$ are empty.}
\label{fig:twoempty}
\end{figure}

As we have proven that it is always possible to construct a matrix~$M'$ such that $\psi(M') = (0,0,c')$ for some $c' < 0$, 
we complete the proof.
\end{proof}

Note that in the above proof, we do not give optimal bounds on number of repetitions of a sequence.

We illustrate Lemma~\ref{lem:nonparallel} in the next example.

\begin{example}
Consider a semigroup $S$ generated by matrices
\begin{align*}
&\begin{pmatrix}
1&-4&20\\0&1&-6\\0&0&1
\end{pmatrix}, &&\begin{pmatrix}
1&3&20\\0&1&-2\\0&0&1
\end{pmatrix}, &&\begin{pmatrix}
1&-1&20\\0&1&1\\0&0&1
\end{pmatrix}, &&\begin{pmatrix}
1&2&20\\0&1&7\\0&0&1
\end{pmatrix}.
\end{align*}
A simple calculation shows that a product of the four matrices (in any order) is a matrix $M$ such that $\psi(M)=(0,0,80+x)$ for some $x\in\mathbb{Z}$. Our goal, is to minimize $x$ by multiplying the matrices in a different order. Denote the given matrices by $M_{({\scriptscriptstyle +,+})}=\begin{psmallmatrix}
1&2&20\\0&1&7\\0&0&1
\end{psmallmatrix}$, $M_{({\scriptscriptstyle +,-})}=\begin{psmallmatrix}
1&3&20\\0&1&-2\\0&0&1
\end{psmallmatrix}$, $M_{({\scriptscriptstyle -,-})}=\begin{psmallmatrix}
1&-4&20\\0&1&-6\\0&0&1
\end{psmallmatrix}$ and $M_{({\scriptscriptstyle -,+})}=\begin{psmallmatrix}
1&-1&20\\0&1&1\\0&0&1
\end{psmallmatrix}$, and 
\begin{align*}
N_1=M_{({\scriptscriptstyle +,+})}M_{({\scriptscriptstyle +,-})}M_{({\scriptscriptstyle -,-})}M_{({\scriptscriptstyle -,+})}=\begin{psmallmatrix}
1&0&47\\0&1&0\\0&0&1
\end{psmallmatrix}.
\end{align*}
That is, $x=-33$. By considering several copies of the product, we can have a negative value in the top right corner. Indeed, consider the product of $16$ matrices 
\begin{align*}
N_2=M_{({\scriptscriptstyle +,+})}^4M_{({\scriptscriptstyle +,-})}^4M_{({\scriptscriptstyle -,-})}^4M_{({\scriptscriptstyle -,+})}^4=\begin{psmallmatrix}
1&0&-22\\0&1&0\\0&0&1
\end{psmallmatrix}.
\end{align*}
Since, we have a matrix with negative value in the top corner, the identity matrix can be generated for example by the product $N_1^{22}N_2^{47}$.
\end{example}

\begin{theorem}\label{thm:ptime}
The identity problem for a semigroup generated by matrices from \heisq is in polynomial time.
\end{theorem}
\begin{proof}
Let $S$ be the matrix semigroup in \heisq generated by the set $G = \{M_1, \ldots, M_r\}$.
There are two possible cases of having the identity matrix in the matrix semigroup 
in \heisq. Either the identity matrix is generated by a product of matrices with pairwise parallel superdiagonal vectors or there are at least two matrices with non-parallel superdiagonal vectors.

Consider the first case. 
Lemma~\ref{lem:single} provides a formula to compute the value in the top corner regardless of the order of the multiplications. That is, we need to solve a system of linear homogeneous Diophantine equations with solutions over non-negative integers.
We partition 
the set $G$ into several disjoint subsets~$G_1, G_2, \ldots,G_s$, where $s$ is at most $r$, and each subset contains matrices with parallel superdiagonal vectors. Since superdiagonal vectors being parallel is a transitive and symmetric property, each matrix needs to be compared to a representative of each subset. If there are no matrices with parallel superdiagonal vectors, then there are $r$ subsets $G_i$ containing exactly one matrix and $O(r^2)$ tests were done.
Let us consider $G_i = \{ M_{k_1}, \ldots, M_{k_{s_i}}\}$, i.e., one of the subsets containing $s_i$ matrices and 
$\psi(M_{k_j}) = (a_{k_j},b_{k_j},c_{k_j})
$. 
By Lemma~\ref{lem:single}, the value 
$c_{k_{j}} - \frac{q_i}{2} a_{k_{j}}^2$,
for a fixed $q_i \in \mathbb{Q}$, is added to the top corner when matrix $M_{k_{j}}$ is multiplied.

We solve the system of two linear homogeneous Diophantine equations~$A \by = \bm{0}$, where
\begin{align*}
A = \begin{pmatrix}
a_{k_1} & a_{k_2}& \cdots & a_{k_{s_i}}\\
c_{k_{1}} - \frac{q_i}{2} a_{k_1}^2 & c_{k_{2}} - \frac{q_i}{2} a_{k_2}^2 & \cdots &  c_{k_{s_i}} -\frac{q_i}{2} a_{k_{s_i}}^2
\end{pmatrix}
\end{align*}
and $\by^\mathsf{T} \in \N^{s_i}$. 
The first row is the constraint that guarantees that the first component of the superdiagonal is zero
in the matrix product constructed from a solution. Since the superdiagonal vectors are parallel, it also  implies that the whole
vector is zero. The second row guarantees that the upper corner is zero.

It is obvious that the identity matrix is in the semigroup if we have a solution in the system of linear homogeneous Diophantine equations for any subset $G_i$. That is, we need to solve at most $r$ systems of two linear homogeneous Diophantine equations.

Next, we consider the second case, where by Lemma~\ref{lem:nonparallel}, it is enough to check whether there exists a sequence of matrices generating a matrix with zero superdiagonal vector and containing two matrices with non-parallel superdiagonal vectors. 
Let us say that $M_{i_1}, M_{i_2} \in G$, where $1 \le i_1,i_2 \le r$ are the two matrices. Recall that $G = \{M_1, M_2, \ldots, M_r\}$ is a generating set 
of the matrix semigroup and let $\psi(M_i) = (a_i, b_i, c_i)$ for all $1 \le i \le r$. 
We can see that there exists such a product containing the two matrices by solving a system of two linear homogeneous Diophantine equations 
of the form $B \by = \bm{0}$, where 
\begin{align*}
B = \begin{pmatrix}
a_{1} & a_{2}& \cdots & a_{r}\\
b_{1} & b_{2}& \cdots & b_{r}
\end{pmatrix},
\end{align*}
with an additional constraint that the numbers in the solution~$\by$ 
that correspond to $M_{i_1}$ and $M_{i_2}$ are non-zero since we must use these two matrices in 
the product. We repeat this process at most $r(r-1)$ times until we find a solution.
Therefore, 
the problem reduces again to solving at most $O(r^2)$ systems of two linear homogeneous Diophantine equations.

Finally, we conclude the proof by mentioning that the identity problem 
for matrix semigroups in the Heisenberg group over rationals \heisq can be decided in polynomial time as, by Lemma~\ref{lem:DiophantineP}, the problem of existence of a positive integer solution to a system of linear homogeneous Diophantine equations is in polynomial time. Note that if the system is non-homogeneous, then solvability of a system of linear Diophantine equations with solutions over positive integers is an $\NP$-complete problem; see for example \cite{Papadimitriou81}. 
\end{proof}

Next, we generalize the above algorithm for the identity problem in the Heisenberg 
group~\heisq to the domain of the Heisenberg groups for 
any dimension over the rational numbers. 
Similarly to the case of dimension three, we establish the following 
result for the case of matrices where multiplication is commutative.

\begin{lemma}\label{lem:single2}
Let $G = \{ M_1, M_2, \ldots, M_r \} \subseteq {\rm H}(n,\mathbb{Q})$ be a set of 
matrices from the Heisenberg group such that $\psi(M_i) = (\ba_i,\bb_i,c_i)$ and $\psi(M_j) = (\ba_j,\bb_j,c_j)$ and $\ba_i \cdot \bb_j = \ba_j \cdot \bb_i$ for any $1 \le i \ne j \le r$. 
If there exists a sequence of matrices
$M = M_{i_1} M_{i_2} \cdots M_{i_k},$
where $i_j \in [1,r]$ for all $1 \le j \le k$, such that $\psi(M) = (\bm{0},\bm{0},c)$ for some $c \in \mathbb{Q}$, 
then, 
\begin{align*}
c = \sum_{j=1}^k (c_{i_j} - \frac{1}{2}  \ba_{i_j} \cdot \bb_{i_j}).
\end{align*}
\end{lemma}

\begin{proof}
Consider the sequence $M_{i_1} M_{i_2} \cdots M_{i_k}$ and let $\psi(M_i) = (\ba_i, \bb_i, c_i)$ for each $i \in [1,r]$. From the multiplication of 
matrices, we have the following equation:
\begin{align*}
c  &=  \sum_{j=1}^k c_{i_j}   +  \sum_{\ell=1}^{k-1} \left( \sum_{j=1}^{\ell}\ba_{i_j} \right) \cdot \bb_{i_{\ell+1}}  
=  \sum_{j=1}^k c_{i_j}   + \frac{1}{2} \left( \sum_{\ell=1}^{k}\sum_{j=1}^{k} \ba_{i_\ell} \cdot \bb_{i_j}  - \sum_{j=1}^k \ba_{i_j} \cdot \bb_{i_j}  \right) \\
 &=  \sum_{j=1}^k (c_{i_j} - \frac{1}{2}  \ba_{i_j} \cdot \bb_{i_j}).
\end{align*}
Note that the first equality follows from a direct computation as in equation~\eqref{eq:sum}. From the above equation, we prove the statement claimed in the lemma. Moreover, due to 
the commutativity of multiplication, the value~$c$ does not change even if we change the 
order of multiplicands.
\end{proof}

Lemma~\ref{lem:nonparallel} does not generalize to \heisnq in the same way as we cannot classify matrices according to types to control the value in upper-right corner, so we use a different technique to prove that the value in the upper corner will be diverging to both positive and negative infinity quadratically as we repeat the same sequence generating any matrix $M$ such that $\psi(M)=(\bm{0},\bm{0},c)$.

\begin{lemma}\label{lem:nonparallel2}
Let $S = \langle M_1, \ldots, M_r \rangle \subseteq {\rm H}(n,\mathbb{Q})$ be a finitely generated matrix semigroup. Then 
the identity matrix exists in $S$ if there exists a sequence of matrices
$M_{i_1} M_{i_2} \cdots M_{i_k},$
where $i_j \in [1,r]$ for all $1 \le j \le k$, satisfying the following properties:
\begin{enumerate}[(i)]
\item $\psi(M_{i_1} M_{i_2} \cdots M_{i_k}) = (\bm{0},\bm{0},c)$ for some $c \in \mathbb{Q}$, and
\item $\ba_{i_{j_1}} \cdot \bb_{i_{j_2}} \ne \ba_{i_{j_2}} \cdot \bb_{i_{j_1}}$ for some $j_1, j_2 \in [1,k]$, where $\psi(M_i) = (\ba_i, \bb_i, c_i)$ for $1 \le i \le r$. 
\end{enumerate}
\end{lemma}

\begin{proof}

From the first property claimed in the lemma, we know that any permutation 
of the sequence of matrix multiplications of $M_{i_1}  \cdots M_{i_k}$ 
results in matrices $M'$ such that $\psi(M') = (\bm{0}, \bm{0}, y)$ for some 
$y \in \mathbb{Q}$ since the multiplication of matrices in the Heisenberg group 
performs additions of vectors which is commutative in the top row and the 
rightmost column excluding the upper-right corner. From the commutative 
behaviour in the horizontal and vertical vectors of matrices in the Heisenberg group, we also know that if we duplicate the matrices in the sequence $M_{i_1} \cdots M_{i_k}$ and multiply the matrices in any order, then the resulting matrix has a non-zero coordinate in the upper triangular coordinates only in the upper right corner.

Now let $j_1, j_2 \in [1,k]$ be two indices such that $\ba_{i_{j_1}} \cdot \bb_{i_{j_2}} 
\ne \ba_{i_{j_2}} \cdot \bb_{i_{j_1}}$ as claimed in the lemma. Then 
consider the following matrix~$M_d$ that can be obtained by duplicating the 
sequence $M_{i_1} \cdots M_{i_k}$ of matrices into $\ell$ copies 
and shuffling the order as follows: $M_d = M_{i_{j_1}}^\ell M_{i_{j_2}}^\ell M_x^\ell,$
where $M_x$ is a matrix that is obtained by multiplying the matrices in 
$M_{i_1}\cdots M_{i_k}$ except the two matrices $M_{j_1}$ and $M_{j_2}$.
Then it is clear that $\psi(M_d) = (\bm{0}, \bm{0}, z)$ for some $z$.
Let us say that $\psi(M_x) = (\ba_x, \bb_x, c_x)$. Then it is 
easy to see that $\ba_{i_{j_1}} + \ba_{i_{j_2}} + \ba_x = \bm{0}$ and 
$\bb_{i_{j_1}} + \bb_{i_{j_2}} + \bb_x = \bm{0}$. Now we show that 
we can always construct two matrices that have only one non-zero rational number in the upper right corner with different signs. 

First, let us consider the $\ell$th power of the matrix $M_{i_{j_1}}$ as 
follows:
\begin{align*}
\!\!\!\!\!\psi(M_{i_{j_1}}^\ell) = (\ba_{i_{j_1}} \ell, \bb_{i_{j_1}} \ell, c_{i_{j_1}} \ell   +  \sum_{h=1}^{\ell -1} h (\ba_{i_{j_1}}  \cdot \bb_{i_{j_1}} ) ) 
=  (\ba_{i_{j_1}} \ell, \bb_{i_{j_1}} \ell, c_{i_{j_1}} \ell   +  \ba_{i_{j_1}}  \cdot \bb_{i_{j_1}} \frac{(\ell-1) \ell}{2}).
\end{align*}
It follows that the matrix~$M_d$ satisfies the equation~$\psi(M_d) = (\bm{0},\bm{0}, z)$ 
such that
\begin{align*}
z  &=  y \ell + (\ba_{i_{j_1}}  \cdot \bb_{i_{j_1}}  + \ba_{i_{j_2}}  \cdot \bb_{i_{j_2}} + \ba_x  \cdot \bb_x ) \frac{(\ell-1) \ell}{2} + (\ba_{i_{j_1}}  \cdot \bb_{i_{j_2}} 
+ (\ba_{i_{j_1}} + \ba_{i_{j_2}})  \cdot \bb_x) \ell^2\\
&= \frac{1}{2}((\ba_{i_{j_1}}  \cdot \bb_{i_{j_1}}  + \ba_{i_{j_2}}  \cdot \bb_{i_{j_2}} + \ba_x  \cdot \bb_x )  + 2 (\ba_{i_{j_1}}  \cdot \bb_{i_{j_2}} 
+ (\ba_{i_{j_1}} + \ba_{i_{j_2}})  \cdot \bb_x)) \ell^2 \\
& \quad {}  + \frac{1}{2}(2y -  (\ba_{i_{j_1}}  \cdot \bb_{i_{j_1}}  + \ba_{i_{j_2}}  \cdot \bb_{i_{j_2}} + \ba_x  \cdot \bb_x ) ) \ell.
\end{align*}
Now the coefficient of the highest term $\ell^2$ in $z$ can be simplified as follows:
\begin{align*}
&\frac{1}{2}((\ba_{i_{j_1}}  \cdot \bb_{i_{j_1}}  + \ba_{i_{j_2}}  \cdot \bb_{i_{j_2}} + \ba_x  \cdot \bb_x )  + 2 (\ba_{i_{j_1}}  \cdot \bb_{i_{j_2}} + (\ba_{i_{j_1}} + \ba_{i_{j_2}})  \cdot \bb_x))\\
&\qquad= \frac{1}{2}((\ba_{i_{j_1}} + \ba_{i_{j_2}}) \cdot (\bb_{i_{j_1}} + \bb_{i_{j_1}}) + \ba_{i_{j_1}}  \cdot \bb_{i_{j_2}} - \ba_{i_{j_2}}  \cdot \bb_{i_{j_1}}    + (\ba_{i_{j_1}} + \ba_{i_{j_2}})  \cdot \bb_x) \\
&\qquad= \frac{1}{2} ((- \ba_x) \cdot (- \bb_x) + \ba_{i_{j_1}}  \cdot \bb_{i_{j_2}} - \ba_{i_{j_2}}  \cdot \bb_{i_{j_1}}   + (- \ba_x) \cdot \bb_x)\\
&\qquad= \frac{1}{2} (\ba_{i_{j_1}}  \cdot \bb_{i_{j_2}} - \ba_{i_{j_2}}  \cdot \bb_{i_{j_1}}).
\end{align*}
By the second property claimed in the lemma, we know that the coefficient of the highest term~$\ell^2$ 
in $z$ cannot be zero. Moreover, the value of $z$ will be diverging to 
negative or positive infinity depending on the sign of $\ba_{i_{j_1}}  \cdot \bb_{i_{j_2}} - \ba_{i_{j_2}}  \cdot \bb_{i_{j_1}}$. Now we consider a different matrix~$M_e$ which is defined to be 
the following product $M_{i_{j_2}}^\ell M_{i_{j_1}}^\ell M_x^\ell$ and say that $\psi(M_e) = (\bm{0}, \bm{0}, e)$ for some $e \in \mathbb{Q}$. Since we have changed the role of two matrices $M_{i_{j_1}}$ and $M_{i_{j_2}}$, the value of $e$ can be represented by a quadratic equation where the coefficient of the highest term is $\ba_{i_{j_2}}  \cdot \bb_{i_{j_1}} - \ba_{i_{j_1}}  \cdot \bb_{i_{j_2}}$.
Therefore, we have proved that it is always possible to construct two matrices that have only one non-zero rational number in the upper right corner with different signs. Then, as in the proof Lemma~\ref{lem:nonparallel}, the identity matrix always exists in the semigroup as we can multiply these two matrices correct number of times to have zero in the upper right coordinate as well. 
\end{proof}

Next, we prove that the identity problem is decidable for $n$-dimensional Heisenberg matrices. In contrast to Theorem~\ref{thm:ptime}, we do not claim that the problem is decidable in polynomial time since one of the steps of the proof is to partition matrices according to dot products which cannot be extended to higher dimensions than three. For higher dimensions, partitioning matrices according to dot products takes an exponential time in the number of matrices in the generating set. Note that if the size of the generating set is fixed, i.e., only the matrices are part of the input, then the problem remains in $\P$. 

\begin{theorem}\label{thm:heisnq}
The identity problem for finitely generated matrix semigroups in the 
Heisenberg group~${\rm H}(n,\mathbb{Q})$ is decidable.
\end{theorem}

\begin{proof}
Similarly to the proof of Theorem~\ref{thm:ptime}, there are two ways the identity matrix can be generated. Either all the matrices commute or there are at least two matrices that do not commute. 

Let $S$ be the matrix semigroup in \heisnq generated by the set $G = \{M_1, M_2, \ldots, M_r\}$.
Consider matrices $N_1,N_2$ and $N_3$, such that $\psi(N_1)=(\ba_1,\bb_1,c_1)$, $\psi(N_2)=(\ba_2,\bb_2,c_2)$ and $\psi(N_3)=(\ba_3,\bb_3,c_3)$. If $\ba_1\cdot \bb_2=\ba_2\cdot\bb_1$ and $\ba_2\cdot \bb_3=\ba_3\cdot \bb_2$, it does not imply that $\ba_1\cdot \bb_3=\ba_3\cdot \bb_1$. Therefore, the number of subsets of $G$, where 
each subset contains matrices that commute with other matrices in the same subset, is exponential 
in $r$ as two different subsets are not necessarily disjoint.

Now we examine whether it is possible to generate the identity matrix by multiplying matrices in each 
subset by Lemma~\ref{lem:single2}. If it is not possible, we need to consider the case of having two matrices 
that do not commute with each other in the product with zero values in the upper-triangular coordinates 
except the corner. Let us say that $M_{i_1}, M_{i_2} \in G$, where $1 \le i_1,i_2 \le r$ are the two matrices. 
Recall that $G = \{M_1, M_2, \ldots, M_r\}$ is a generating set 
of the matrix semigroup and let $\psi(M_i) = (\ba_i, \bb_i, c_i)$ for all $1 \le i \le r$. We also denote the $m$th element 
of the vector~$\ba_i$ (respectively, $\bb_i$) by $\ba_i[m]$ (respectively, $\bb_i[m]$) for $1 \le m \le n-2$.

Then we can see that there exists such a product by solving a system of $2(n-2)$ linear homogeneous Diophantine equations 
of the form $B \by = \bm{0}$, where
\begin{align*}
B = \begin{pmatrix}
\ba_1[1] & \cdots & \ba_r[1] \\
\vdots & \ddots & \vdots \\
\ba_1[d-2] & \cdots & \ba_r[d-2] \\
\bb_1[1] & \cdots & \bb_r[1] \\
\vdots &  \ddots & \vdots \\
\bb_1[d-2] & \cdots & \bb_r[d-2] \\
\end{pmatrix},
\end{align*}
with an additional constraint that the numbers in the solution~$\by$ 
that correspond to $M_{i_1}$ and $M_{i_2}$ are non-zero since we must use these two matrices in 
the product. We repeat this process at most $r(r-1)$ times until we find a solution.

Hence, we can view the identity problem in 
${\rm H}(n,\mathbb{Q})$ for $n \ge 3$ as the problem of solving systems of 
$2(n-2)$ linear homogeneous Diophantine equations with some constraints on the solution. 
By Lemma~\ref{lem:DiophantineP}, we can solve systems of linear homogeneous Diophantine equations in polynomial time, thus we conclude that the identity problem 
in ${\rm H}(n,\mathbb{Q})$ is also decidable. 
\end{proof}

\section{The identity problem in matrix semigroups in dimension four}\label{sec:IPfour}
In this section, we prove that the identity problem is undecidable for $4\times4$ matrices, when the generating set has eight matrices, by introducing a new technique exploiting the anti-diagonal entries.

\begin{theorem}\label{thm:identity4}
Given a semigroup $S$ generated by eight $4 \times 4$ 
integer matrices with determinant one, determining whether the identity matrix belongs
to $S$ is undecidable. 
\end{theorem}

\begin{proof}
We prove the claim by reducing from the \PCP. 
We shall use an encoding to embed an instance of the \PCP into 
a set of $4 \times 4$ integer matrices.

Let $\alpha$ be the mapping of Lemma~\ref{lem:groupEnc} that maps elements of arbitrary group alphabet into a binary group alphabet $\Gamma_2=\{a,b,\overbar{a},\overbar{b}\}$.
We also define a monomorphism $f : \FG[\Gamma_2] \to \mathbb{Z}^{2\times 2}$ as $f(a) = \begin{psmallmatrix}1 & 2 \\ 0 & 1 \end{psmallmatrix}$, $f(\overbar{a}) = \begin{psmallmatrix}1 & -2 \\ 0 & 1 \end{psmallmatrix}$, $f(b) = \begin{psmallmatrix}1 & 0 \\ 2 &1 \end{psmallmatrix}$ and $f(\overbar{b}) = \begin{psmallmatrix}1 & 0 \\ -2 & 1 \end{psmallmatrix}$.
Recall that the matrices $\begin{psmallmatrix}1 & 2 \\ 0 & 1 \end{psmallmatrix}$ and $\begin{psmallmatrix}1 & 0 \\ 2 &1 \end{psmallmatrix}$ generate a free subgroup of \sltwoz~\cite{LS77}.
The composition of two monomorphisms~$\alpha$ and $f$
gives us the embedding from an arbitrary group alphabet into the 
special linear group~\sltwoz. 
We use the composition of two monomorphisms~$\alpha$ and $f$ to encode a set of pairs of words over an arbitrary group alphabet into a 
set of $4 \times 4$ integer matrices in \slfourz and denote it by $\beta$.

Let $(g,h)$ be an instance of the \PCP, where $g,h:\{a_1,\ldots,a_n\}^*\to \Sigma_2^*$, where $\Sigma_2=\{a,b\}$. Without loss of generality, we can assume that the solution starts with the letter $a_1$. Moreover, we assume that this is the only occurence of $a_1$.
We define the alphabet $\Gamma = \Sigma_2\cup \Sigma_2^{-1} \cup \Sigma_B\cup \Sigma_B^{-1}$, where $\Sigma_B = \{q_0, q_1, p_0,p_1\}$ is the alphabet for the border letters that enforce the form of a solution.
%

Let us define the following sets of words 
$W_1 \cup W_2 \subseteq \FG \times \FG$, where 
\begin{align*}
W_1 &= \left\{ (q_0 a\overbar{q_0},p_0a\overbar{p_0}),\ (q_0 b\overbar{q_0},p_0b\overbar{p_0}) \mid a,b \in \Sigma_2,\;\; q_0,p_0 \in \Sigma_B \right\} \text{ and} \\
W_2 &= \left\{ (q_0\overbar{g(a_1)}\overbar{q_1}, p_0\overbar{h(a_1)}\overbar{p_1}),\ (q_1\overbar{g(a_i)}\overbar{q_1}, p_1\overbar{h(a_i)}\overbar{p_1})\mid 1 < i \le n, \;\; q_0,q_1,p_0,p_1 \in \Sigma_B \right\}.
\end{align*}
Intuitively, the words from set $W_1$ are used to construct words over $\Sigma_2$ and the words from set $W_2$ to cancel them according to the instance of the \PCP.

Let us prove that $(q_0 \overbar{q_1}, p_0 \overbar{p_1}) \in \FG[W_1 \cup W_2]$ if and only 
if the \PCP has a solution. It is easy to see that any pair of non-empty words in $\FG[W_1]$ is 
of the form~$(q_0 w \overbar{q_0}, p_0 w \overbar{p_0})$ for $w \in \Sigma_2^+$. Then there exists a pair of words in $\FG[W_2]$ of the form 
$(q_0 \overbar{w} \overbar{q_1}, p_0 \overbar{w}\overbar{p_1})$ for some word~$w \in \Sigma_2^+$ if and only if the \PCP has a solution. Therefore, the pair 
of words $(q_0 \overbar{q_1}, p_0 \overbar{p_1})$ can be constructed 
by concatenating pairs of words in $W_1$ and $W_2$ if and only if the \PCP has a solution.

For each pair of words~$(u, v) \in \FG[W_1 \cup W_2]$, we define a matrix~$A_{u,v} 
$ to be $\begin{psmallmatrix} \beta(u) & \bm{0}_2 \\ \bm{0}_2 & \beta(v) \end{psmallmatrix} \in {\rm SL}(4,\mathbb{Z}),$
where $\bm{0}_2$ is the zero matrix in $\mathbb{Z}^{2 \times 2}$. Moreover, we 
define the following matrix 
\begin{align*}
B_{q_1 \overbar{q_0}, p_1 \overbar{p_0}} = \begin{pmatrix} \bm{0}_2 & \beta(q_1 \overbar{q_0}) \\ \beta(p_1 \overbar{p_0}) & \bm{0}_2 \end{pmatrix} \in {\rm SL}(4,\mathbb{Z}).
\end{align*}

Let $S$ be a matrix semigroup generated by the set
$
\{ A_{u,v}, B_{q_1 \overbar{q_0}, p_1 \overbar{p_0}} \mid (u,v) \in W_1 \cup W_2 \}.
$
We already know that the pair $(q_0 \overbar{q_1}, p_0 \overbar{p_1})$ 
of words can be generated by concatenating words in $W_1$ and $W_2$ 
if and only if the \PCP has a solution. The matrix semigroup $S$ has the 
corresponding matrix~$A_{q_0 \overbar{q_1}, p_0 \overbar{p_1}}$ and 
thus, 
\begin{align*}
\begin{pmatrix} \beta(q_0 \overbar{q_1}) & \bm{0}_2 \\ \bm{0}_2 & \beta(p_0 \overbar{p_1}) \end{pmatrix} \begin{pmatrix} \bm{0}_2 & \beta(q_1 \overbar{q_0}) \\ \beta(p_1 \overbar{p_0}) & \bm{0}_2 \end{pmatrix} = 
\begin{pmatrix}
\bm{0}_2 & \beta(\varepsilon) \\ \beta(\varepsilon) & \bm{0}_2
\end{pmatrix} \in S.
\end{align*}

Then we see that the identity matrix~$\bm{I}_4$ exists in the semigroup~$S$ as follows:
\begin{align*}
\begin{pmatrix}
\bm{0}_2 & \beta(\varepsilon) \\ \beta(\varepsilon) & \bm{0}_2
\end{pmatrix}
\begin{pmatrix}
\bm{0}_2 & \beta(\varepsilon) \\ \beta(\varepsilon) & \bm{0}_2
\end{pmatrix} = \begin{pmatrix}
\beta(\varepsilon) & \bm{0}_2  \\ \bm{0}_2 & \beta(\varepsilon) 
\end{pmatrix} = \begin{pmatrix}
\bm{I}_2 & \bm{0}_2  \\ \bm{0}_2 & \bm{I}_2
\end{pmatrix} =  \bm{I}_4 \in S.
\end{align*}

Now we prove that the identity matrix does not exist in $S$ if the \PCP 
has no solution. It is easy to see that we cannot obtain the identity matrix only
by multiplying `$A$' matrices since there is no possibility of cancelling 
every border letter. We need to multiply the matrix~$B_{q_1 \overbar{q_0}, p_1 \overbar{p_0}}$
with a product of `$A$' matrices at some point to reach the identity matrix. 
Note that the matrix~$B_{q_1 \overbar{q_0}, p_1 \overbar{p_0}}$ cannot be the first matrix of the product, followed by the `$A$' matrices, because the upper right block of $B_{q_1 \overbar{q_0}, p_1 \overbar{p_0}}$, 
which corresponds to the first word of the pair, should be multiplied with the lower right block of `$A$' 
matrix, which corresponds to the second word of the pair. 

Suppose that the `$A$' matrix is of form $\begin{psmallmatrix} \beta(q_0 u \overbar{q_1}) & \bm{0}_2 \\ \bm{0}_2 & \beta(p_0 v \overbar{p_1}) \end{psmallmatrix}$. 
Since the \PCP instance has no solution, either $u$ or $v$ is not the empty 
word. We multiply $B_{q_1 \overbar{q_0}, p_1 \overbar{p_0}}$ to the matrix 
and then obtain the following matrix:
{\small\begin{align*}
\begin{pmatrix} \beta(q_0 u \overbar{q_1}) & \bm{0}_2 \\ \bm{0}_2 & \beta(p_0 v \overbar{p_1}) \end{pmatrix}
\begin{pmatrix} \bm{0}_2 & \beta(q_1 \overbar{q_0}) \\ \beta(p_1 \overbar{p_0}) & \bm{0}_2 \end{pmatrix}
= \begin{pmatrix}
\bm{0}_2 & \beta(q_0 u \overbar{q_0}) \\ \beta(p_0 v \overbar{p_0}) & \bm{0}_2
\end{pmatrix}.
\end{align*}}
We can see that either the upper right part or the lower left part cannot be 
$\beta(\varepsilon)$, which actually corresponds to the identity matrix in $\mathbb{Z}^{2\times 2}$. 
Now the only possibility of reaching the identity matrix is to multiply matrices which 
have ${\rm SL}(2,\mathbb{Z})$ matrices in the anti-diagonal coordinates like $B_{q_1 \overbar{q_0}, p_1 \overbar{p_0}}$. However, we cannot cancel the parts because the upper right block (the lower left block)
of the left matrix is multiplied with the lower left block (the upper right block) of the right matrix 
as follows:
{\small\begin{align*}
\begin{pmatrix}
\bm{0}_2 & A \\ B & \bm{0}_2
\end{pmatrix}
\begin{pmatrix}
\bm{0}_2 & C \\ D & \bm{0}_2
\end{pmatrix} 
= 
\begin{pmatrix}
AD & \bm{0}_2 \\ \bm{0}_2 & BC
\end{pmatrix},
\end{align*}}
where $A,B,C$ and $D$ are matrices in $\mathbb{Z}^{2 \times 2}$.
As the first word of the pair is encoded in the upper right block of the matrix and 
the second word is encoded in the lower left block, it is not difficult to see that 
we cannot cancel the remaining 
blocks.
%
%

Currently, the undecidability bound for the \PCP is five \cite{Neary15} and thus the semigroup $S$ is generated by eight matrices. Recall that in the beginning of the proof, we assumed that the letter $a_1$ of the \PCP is used exacly once and is the first letter of a solution. This property is in fact present in \cite{Neary15}.
\end{proof}

Consider 
the membership problem called the {\em special diagonal membership problem}, 
where the task is to determine whether a scalar multiple of the identity matrix exists in a 
given matrix semigroup. 
The most recent undecidability bound in $\Z^{4\times4}$ is shown to be 14 by 
Halava et al.~\cite{HHH07}. 
We improve the bound to eight, as
the identity matrix is the only diagonal matrix of the semigroup $S$ in the proof of Theorem~\ref{thm:identity4}.
We also prove that the identity problem is undecidable 
in \HQS as well by replacing the composition $f \circ \alpha$ of mappings 
with a mapping from a group alphabet to the set of rational 
quaternions; see \cite{BP08}.

\begin{corollary}
For a given semigroup $S$ generated by eight $4 \times 4$ 
integer matrices, determining whether there exists any diagonal 
matrix in $S$ is undecidable. 
\end{corollary}


\begin{corollary}
For a given semigroup $S$ generated by eight $2 \times 2$ 
rational quaternion matrices, determining whether there exists the identity
matrix in $S$ is undecidable.
\end{corollary}

\section{Concluding remarks}
In this paper, we considered the identity problem in matrix semigroups and provided a better bound on the number of matrices (reducing from 48 to 8) in the generator set for $4\times4$ integer matrices, where the problem is undecidable. More importantly, we showed that there is no embedding of pairs of words into ${\rm SL}(3,\Z)$. While this does not imply that the identity problem is decidable, it does provide strong evidence about decidability of computational problems in \slthreez. 
Then we showed that the identity problem is decidable for Heisenberg group \heis, which is an important subgroup of ${\rm SL}(3,\mathbb{Z})$, and generalized the result to \heisnq for any $n\in\N$. The natural follow-up question is whether other standard matrix problems, such as membership, are decidable in $\heis$ or whether the identity problem is decidable for ${\rm H}(3,\mathbb{C})$.

\bibliographystyle{plainurl}
\bibliography{identity}

\end{document}